\documentclass[11pt]{article}
\pdfoutput=1
\usepackage{amsmath,amssymb,amsthm,amsfonts,latexsym,mathtools,mathdots,graphicx,float,bbold}
\usepackage{caption,subcaption,ellipsis,xcolor,textcomp,thmtools,thm-restate,cmap,mleftright,framed,nicematrix}
\usepackage[backref=page,colorlinks,citecolor=blue,bookmarks=true]{hyperref}
\usepackage[nameinlink]{cleveref}
\crefname{ineq}{inequality}{inequalities}
\creflabelformat{ineq}{#2{\upshape(#1)}#3}
\crefname{fact}{fact}{facts}
\creflabelformat{fact}{#2#1#3}
\crefname{obs}{observation}{observations}
\creflabelformat{obs}{#2{\upshape(#1)}#3}
\usepackage[letterpaper,margin=1in]{geometry}
\usepackage{enumitem}

\declaretheorem[numberwithin=section]{theorem}
\declaretheorem[numbered=no,name=Theorem]{theorem*}
\declaretheorem[sibling=theorem]{lemma}
\declaretheorem[sibling=theorem]{claim}
\declaretheorem[sibling=theorem]{proposition}
\declaretheorem[sibling=theorem]{corollary}
\declaretheorem[sibling=theorem]{conjecture}

\theoremstyle{definition}
\declaretheorem[sibling=theorem]{definition}
\declaretheorem[sibling=theorem]{fact}

\declaretheorem[sibling=theorem]{problem}

\let\Pr\relax
\DeclareMathOperator*{\Pr}{\mathbf{Pr}}
\DeclareMathOperator*{\E}{\mathbf{E}}

\DeclareMathOperator*{\Var}{\mathbf{Var}}

\DeclarePairedDelimiter{\ket}{\lvert}{\rangle}\DeclarePairedDelimiterX\innerp[2]{\langle}{\rangle}{#1\delimsize\vert\mathopen{}#2}\DeclarePairedDelimiterX\braket[2]{\langle}{\rangle}{#1\delimsize\vert\mathopen{}#2}\DeclarePairedDelimiterX\braketOP[3]{\langle}{\rangle}{#1\,\delimsize\vert\,\mathopen{}#2\,\delimsize\vert\,\mathopen{}#3}\DeclarePairedDelimiterX\ketbra[2]{\lvert}{\rvert}{#1\delimsize\rangle\!\delimsize\langle#2}\DeclarePairedDelimiterX\outerp[2]{\lvert}{\rvert}{#1\delimsize\rangle\!\delimsize\langle#2}\DeclarePairedDelimiterX\projector[1]{\lvert}{\rvert}{#1\delimsize\rangle\!\delimsize\langle#1}

\DeclareMathOperator{\poly}{poly}
\DeclareMathOperator{\polylog}{polylog}

\DeclareMathOperator{\codim}{codim}
\DeclareMathOperator{\diag}{diag}

\newcommand{\F}{\mathbb{F}}

\newcommand{\eps}{\epsilon}

\newcommand{\calA}{\mathcal{A}}
\newcommand{\calB}{\mathcal{B}}
\newcommand{\calC}{\mathcal{C}}

\newcommand{\calF}{\mathcal{F}}

\newcommand{\calK}{\mathcal{K}}

\newcommand{\calM}{\mathcal{M}}

\newcommand{\calT}{\mathcal{T}}

\newcommand{\calX}{\mathcal{X}}

\newcommand{\bone}{\boldsymbol{1}}

\newcommand{\br}{\boldsymbol{r}}

\DeclarePairedDelimiter\norm{\lVert}{\rVert}

\usepackage{booktabs,array,microtype}
\usepackage{pgfplots}
\usepgfplotslibrary{groupplots}
\pgfplotsset{compat=1.18}
\allowdisplaybreaks

\crefname{problem}{problem}{problems}
\Crefname{problem}{Problem}{Problems}

\DeclareMathOperator{\rank}{rank}

\DeclareMathOperator{\mult}{mult}
\DeclareMathOperator{\red}{red}
\DeclareMathOperator{\ev}{ev}

\newcommand{\G}{\mathbb{G}}
\newcommand{\indicator}{\mathsf{1}}
\newcommand{\Dec}{\mathsf{Dec}}

\newcommand{\Fsf}{\mathsf{F}}
\newcommand{\complexi}{\mathsf{i}}

\newcommand{\Tbb}{\mathbb{T}}
\newcommand{\rtsf}{\mathsf{rt}}
\newcommand{\ksf}{\mathsf{k}}
\newcommand{\Solve}{\mathsf{Solve}}

\newcommand{\Isf}{\mathsf{I}}
\newcommand{\Lsf}{\mathsf{L}}
\newcommand{\Xsf}{\mathsf{X}}
\newcommand{\Csf}{\mathsf{C}}
\newcommand{\Hsf}{\mathsf{H}}

\newcommand{\SubsetSum}{\textnormal{\textsc{Subset-Sum}}\xspace}
\newcommand{\FSubsetSum}{$\F_3^n$\textnormal{\textsc{-}}\SubsetSum}
\newcommand{\FSignedSum}{$\F_3^n$\textnormal{\textsc{-Signed-Sum}}\xspace}

\newcommand{\FqSubsetSum}{$\F_q^n$\textnormal{\textsc{-}}\SubsetSum}
\newcommand{\BinaryErrorLWE}{\textnormal{\textsc{Binary-Error LWE}}\xspace}
\newcommand{\PointwiseBinaryErrorLWE}{\textnormal{\textsc{Pointwise Binary-Error LWE}}\xspace}

\newcounter{boxeddefinition}[section]
\renewcommand{\theboxeddefinition}{\thesection.\arabic{boxeddefinition}}

\crefname{boxeddefinition}{definition}{definitions}
\Crefname{boxeddefinition}{Definition}{Definitions}

\title{Exponential quantum speedup for $\mathbb{F}_3^n$-Subset-Sum? \\ Or, rigorous classical algorithms for Binary-Error LWE}
\author{
Robin Kothari\textsuperscript{1}\quad
Tony Metger\textsuperscript{2}\quad
Ryan O'Donnell\textsuperscript{3}\quad
Noah Shutty\textsuperscript{1}\quad
Kewen Wu\textsuperscript{4}
\\[1em]
\small \textsuperscript{1}Google Quantum AI
\\
\small \textsuperscript{2}Courant Institute of Mathematical Sciences and Department of Physics, New York University
\\
\small \textsuperscript{3}Computer Science Department, Carnegie Mellon University
\\
\small \textsuperscript{4}Department of Computing and Mathematical Sciences, California Institute of Technology
}
\date{}

\begin{document}
\maketitle

\begin{abstract}

We study vector subset sum over $\mathbb{F}_3^n$: given $m$ random vectors from $\mathbb{F}_3^n$, find a nonempty subset that sums to zero; the smaller $m$, the more difficult it is to find such a subset.
Chen, Liu, and Zhandry (EUROCRYPT'22) introduced an efficient quantum algorithm that solves this problem when $m\approx n^2/2$, where a naive classical algorithm would require exponential time.
Subsequently, Kothari, O'Donnell, and Wu (STOC'2026) gave an efficient classical algorithm that only requires $m \approx n^2/3$ vectors, thus removing the hope for an exponential quantum advantage in this parameter regime.

Using the framework of Chen, Liu, and Zhandry, we give quantum algorithms that require much fewer input vectors, renewing the possibility of an exponential quantum speedup: for any fixed $\epsilon>0$, our quantum algorithm solves $\mathbb{F}_3$-subset sum in polynomial time with $m=\epsilon\cdot n^2$ vectors.
More generally, we establish a full sample--time tradeoff that interpolates between exponential and polynomial runtime.

The main ingredient is a deterministic classical algorithm for the binary-error Learning-with-Errors problem, which is of independent cryptographic interest.
For this, we rigorously establish a sample--time tradeoff that was predicted by earlier algebraic heuristics.
For vector subset sums over larger fields, we also significantly improve classical algorithms in Kothari, O'Donnell, and Wu (STOC'2026).
\end{abstract}

\newpage
\tableofcontents
\newpage

\section{Introduction}\label{sec:intro}

Finding natural computational problems that admit exponential quantum
speedups is a central goal in quantum computing. Despite decades of
progress, relatively few families of such candidates are known beyond
quantum simulation and problems with hidden algebraic structure. Discovering new candidates would both
broaden the scope of quantum algorithms and help identify which features
of a problem make a large quantum advantage possible. Search problems
with explicit classical inputs and efficiently checkable solutions are
particularly appealing: they give concrete targets for comparing quantum
and classical methods, even when proving classical lower bounds remains
out of reach.

A particularly interesting class of search problems are those related to post-quantum cryptography.
From the practical side, studying problems adjacent to post-quantum cryptography gives us a clearer understanding of the potential for quantum attacks on such schemes, and even relatively minor quantum speedups or quantum speedups in restricted parameter regimes may be important for the deployment of post-quantum cryptography.
From the theoretical side, Regev's reduction~\cite{Reg09}, which was originally introduced to argue for the \emph{quantum hardness} of the Learning With Errors (LWE) problem, has recently been identified as a potential source of \emph{quantum speedups} in related problems~\cite{CLZ22,yamakawa2024verifiable,DQI}.

In this work, we focus on the \FSubsetSum problem:
given uniformly random vectors $h_1,\ldots,h_m\in\F_3^n$, find a nonempty subset whose sum is zero.
This problem is related to the short-integer-solution (SIS) problem in lattice cryptography, which asks for a ``short'' linear combination of a given set of vectors that sums to zero.
In standard SIS, ``short'' means the Euclidean norm of the coefficient vector is small.
\FSubsetSum is similar, except that ``short'' means that the coefficients are only allowed to be $0$ or $1$ (not $2$).

The metric of interest for the \FSubsetSum is the number of vectors $m$ (as a function of the dimension $n$) that an algorithm requires to efficiently find a subset that sums to zero.
Chen, Liu, and Zhandry~\cite{CLZ22} showed how to use Regev's reduction to construct an efficient quantum algorithm for \FSubsetSum for $m \approx n^2/2$, which was a potential exponential quantum speedup. Kothari, O'Donnell, and Wu~\cite{KOW26} subsequently gave classical algorithms
covering the proposed regimes; in particular, their polynomial-time algorithm for \FSubsetSum
uses $m\approx n^2/3$ vectors.

The use of Regev's reduction in~\cite{CLZ22} reduces \FSubsetSum into a classical decoding problem known as binary-error LWE,
which also has an independent cryptographic motivation (see \Cref{sec:quantum} for details of the reduction).
In binary-error LWE over a general field $\mathbb{F}_q$, one needs to recover
a secret from random linear equations perturbed by errors in $\{0,1\}$.
Binary errors allow simple and efficient constructions, including
proposed encryption schemes for constrained devices~\cite{BGGOP16,STA20}.
For suitable growing moduli and restricted sample counts, Micciancio and
Peikert~\cite{MP13} established reductions from standard LWE to binary-error
LWE. On the algorithmic side, Arora and Ge~\cite{AG11} showed how small
error sets can be exploited by solving polynomial equations, giving
efficient attacks when sufficiently many samples are available.
Understanding the sample--time tradeoff is therefore essential to
assessing the hardness of binary-error LWE and the parameters in which
it can support cryptography.
We remark that prior algebraic analyses predict a precise sample--time tradeoff for certain attacks, but rely on regularity assumptions~\cite{ACFP14,STA20}, which remain unproven.

\paragraph{Our contribution.}
We provide a deterministic classical attack for binary-error LWE and prove an unconditional, complete sample--time tradeoff that matches prior predictions.
Through the
quantum reduction of~\cite{CLZ22}, this yields a polynomial-time quantum algorithm for \FSubsetSum when $m=\eps\cdot n^2$ at every fixed $\eps>0$. 

Recall that the prior work~\cite{KOW26} provides an efficient classical algorithm in the regime $\eps \geq 1/3$ (up to asymptotically vanishing corrections).
To more accurately assess the prospects for an exponential quantum advantage for \FSubsetSum, we also improve the classical algorithms for this problem, achieving a constant $\eps = 0.1022$.
Thus, our result suggests that there might an exponential quantum speedup at values of $\eps$ smaller than this.
We note that this value of $\eps$ is the result of an LLM-assisted optimization of existing classical algorithmic ideas.
Naturally, we do not know (or expect) that this constant is optimal, but achieving a similar sample--time tradeoff as we do in the quantum case (and in particular a classical algorithm for any fixed $\eps > 0$) seems to require substantially new ideas.

We also extend some of our results to larger field sizes. These results are stated in the following subsections, though we mostly focus on the $\mathbb{F}_3$-case for concreteness.

\paragraph{Notation.}
Throughout,
$\log$ denotes the base-two logarithm and $[n]=\{1,2,\ldots,n\}$ for each
positive integer $n$. Asymptotic notation hides absolute constants unless
subscripts indicate dependence on fixed parameters; for example,
$O_{q,\kappa}(\cdot)$ permits dependence on $q$ and $\kappa$.

\subsection{\texorpdfstring{Quantum algorithms for average-case \FSubsetSum}
{Quantum algorithms for average-case F3n-Subset-Sum}}\label{sec:intro_f3_quantum}

We begin by formally stating the \FSubsetSum problem.

\begin{problem}[Average-case \FSubsetSum]\label{prob:f3-subset-sum}
Given input vectors $h_1,\ldots,h_m\in\F_3^n$ uniformly at random, the task is to find a nonempty set $S\subseteq[m]$ such that $\sum_{i\in S}h_i=0$.
\end{problem}

The crucial restriction is that we require an unweighted sum of vectors $h_i$, i.e., we demand a linear combination of the vectors $h_i$ with coefficients in $\{0,1\}$; if we also allowed the coefficient 2, the problem would easily by solved by Gaussian elimination.

Naturally, when $m$ is too small compared to $n$, it is likely that there is no subset that sums to zero. 
Once $m$ becomes large enough, a solution will exist with high probability; for even larger $m$, a solution will exist for all choices of $h_i$, though of course if may be intractable to find.
These information-theoretic properties of the \FSubsetSum problem are collected in the following \Cref{fct:f3-subset-sum-existence}, whose proof is in \Cref{sec:fct:f3-subset-sum-existence}.

\begin{restatable}{fact}{fctfsubsetsumexistence}\label{fct:f3-subset-sum-existence}
Let $\kappa_\star=\log(3)\approx1.5850$ and let $\eta>0$ be any constant. 
For \Cref{prob:f3-subset-sum}, 
\begin{itemize}
\item if $m\ge2n+1$, a solution $S$ always exists;
\item if $m\ge(\kappa_\star+\eta)\cdot n$, a solution $S$ exists with probability at least $1-2^{-\Omega_\eta(n)}$; 
\item if $m\le(\kappa_\star-\eta)\cdot n$, a solution $S$ exists with probability at most $2^{-\Omega_\eta(n)}$.
\end{itemize}
\end{restatable}

We are interested in the computational difficulty of solving \Cref{prob:f3-subset-sum} as a function of the number of vectors $m$.
Prior to our work, the best (classical or quantum) algorithm achieved $m\approx n^2/3$ \cite{KOW26}.
Our main result is a quantum algorithm for \Cref{prob:f3-subset-sum} that smoothly interpolates between the information-theoretic threshold $m\approx\kappa_\star\cdot n$ and the prior polynomial-time-solvable threshold $m\approx n^2/3$. In particular, this yields a polynomial-time quantum algorithm for $m = \epsilon \cdot n^2$ for any fixed $\epsilon$.

\begin{restatable}{theorem}{quantummain}\label{thm:quantum-main}
Let $\kappa>\log(3)$ be any constant.
For every $m\ge\kappa\cdot n$, there is a quantum algorithm that solves \Cref{prob:f3-subset-sum} with probability at least $1-2^{-\Omega_\kappa(m)}$ and runs in time $(m/n)^{O_\kappa(\lceil n^2/m\rceil)}$.
\end{restatable}

Making a few representative parameter choices, we specialize \Cref{thm:quantum-main} as follows.

\begin{corollary}\label{cor:quantum-main}
\Cref{prob:f3-subset-sum} can be solved in quantum
\begin{itemize}
    \item polynomial $\poly(n)$ time for every $m\ge\Omega(n^2)$,
    \item quasipolynomial $2^{\polylog(n)}$ time for every $m\ge n^2/\polylog(n)$,
    \item truly-subexponential $2^{O_\alpha(n^{1-\alpha})}$ time for every $m\ge n^{1+\alpha}\cdot\log n$, where $0<\alpha<1$ is a constant.
\end{itemize}
\end{corollary}

We briefly compare this to the best known classical algorithms.
Write $r=m/n$.  
If $r=\Theta(1)$, the known attacks remain exponential~\cite{BCDL19,CDE21,KL22,WANT26}.
If $r\to\infty$, the prior Wagner-type algorithms achieve runtime $2^{O(n/\log r)}$~\cite{Wagner02,BCDL19,CDE21}.
By contrast, our \Cref{thm:quantum-main} has superior runtime $2^{O(\lceil n/r\rceil\log r)}$: for example, if $r=n^\alpha\log n$ and $m=n^{1+\alpha}\log n$, our runtime is $2^{O_\alpha(n^{1-\alpha})}$ whereas theirs are $2^{O_\alpha(n/\log n)}$.
For polynomial-time solvers, the smallest prior threshold was $m\ge(1/3+o(1))\cdot n^2$~\cite{KOW26}, building on earlier works~\cite{ISS12,IPS18,CLZ22}. In comparison,
\Cref{cor:quantum-main} covers $m\ge\eps\cdot n^2$ for every constant $\eps>0$.

More loosely related are Wagner analyses for other short-relation
problems~\cite{DEL25} and cryptographic assumptions~\cite{Wave19,Wave23,BBWW26}.
Since none of the latter has exactly the distribution and parameter regime of
\Cref{prob:f3-subset-sum}, our asymptotic result alone is not a concrete attack
on them.

\subsection{\texorpdfstring{Classical algorithms for \BinaryErrorLWE over $\F_q$}{Classical algorithms for Binary-Error LWE over Fq}}\label{sec:intro-lwe}

To achieve the results in \Cref{thm:quantum-main}, we use the framework of~\cite{CLZ22}, which makes use of Regev's reduction~\cite{Reg09} to reduce \FSubsetSum to \BinaryErrorLWE. 
In particular, any improved classical solver, which we also call a decoder, for \BinaryErrorLWE yields an improved quantum algorithm for \FSubsetSum.
In this subsection, we explain our decoder for \BinaryErrorLWE and state the rigorous bounds we prove.

As noted, \BinaryErrorLWE is also of independent cryptographic interest over general fields $\F_q$. Our decoder satisfies the following pointwise recovery guarantee.

\begin{problem}[\BinaryErrorLWE over $\F_q$ with a pointwise recovery guarantee]\label{prob:pointwise-binary-error-lwe}
Let $q\ge3$ be a fixed prime. The input is $(A,b)$, where
$A\in\F_q^{n\times m}$ is uniformly random and $b=A^\top s+e$, for
arbitrary unknown vectors $s\in\F_q^n$ and $e\in\{0,1\}^m$ chosen
independently of $A$. The task is to recover the planted secret $s$.

A deterministic decoder $\Dec$ has \emph{pointwise failure probability}
at most $\delta$ if, for every $s\in\F_q^n$ and $e\in\{0,1\}^m$,
\[
 \Pr_{A\sim\F_q^{n\times m}}
 \left[\Dec(A,A^\top s+e)\ne s\right]\le\delta.
\]
The decoder receives only $(A,b)$, and the same decoder must satisfy this
guarantee for all $s,e$.
\end{problem}

We focus on this more general problem, and our results for \FSubsetSum follow as a special case for $\mathbb{F}_3$. 
Similarly to \FSubsetSum, it is easy to establish a uniqueness threshold for \Cref{prob:pointwise-binary-error-lwe}. This is summarized in \Cref{fct:lwe-existence}, which we prove in \Cref{sec:fct:lwe-existence}.

\begin{restatable}{fact}{fctlweexistence}\label{fct:lwe-existence}
Let $\kappa_q=\log_{q/2}(q)$ and let $\eta>0$ be any constant.
For \Cref{prob:pointwise-binary-error-lwe}, if
$m\ge(\kappa_q+\eta)\cdot n$, a unique solution $s$ exists with probability
at least $1-2^{-\Omega_{q,\eta}(m)}$; if
$m\le(\kappa_q-\eta)\cdot n$, a unique solution $s$ exists with probability
at most $2^{-\Omega_{q,\eta}(n)}$.
\end{restatable}

For \Cref{prob:pointwise-binary-error-lwe}, we provide a deterministic classical algorithm whose runtime scales favorably with the number of input vectors.
In the usual \BinaryErrorLWE, the error vector $e$ is sampled from a specified distribution, and the secret $s$ may also be random~\cite{MP13}.
Our decoder's success guarantee holds for every secret and binary error vector, with probability taken only over the random matrix $A$.
The matrices $A$ on which the decoder fails may depend on $(s,e)$.
Averaging this guarantee over uniformly random $s,e$ gives the average-case guarantee that suffices for our quantum reduction.

\begin{restatable}{theorem}{lweuniform}
\label{thm:lwe-uniform}
Let $\kappa>\log_{q/2}(q)$ be any constant.
For every $m\ge\kappa\cdot n$, there is a deterministic classical algorithm that solves
\Cref{prob:pointwise-binary-error-lwe} with probability at least
$1-2^{-\Omega_{q,\kappa}(m)}$ and runs in time
$(m/n)^{O_{q,\kappa}(\lceil n^2/m\rceil)}$.
\end{restatable}

We also specialize \Cref{thm:lwe-uniform} with a few parameter choices.

\begin{corollary}\label{cor:lwe-uniform}
For constant $q\ge3$, \Cref{prob:pointwise-binary-error-lwe} can be solved classically in
\begin{itemize}
    \item polynomial $\poly(n)$ time for every $m\ge\Omega(n^2)$,
    \item quasipolynomial $2^{\polylog(n)}$ time for every $m\ge n^2/\polylog(n)$,
    \item truly-subexponential $2^{O_\alpha(n^{1-\alpha})}$ time for every
    $m\ge n^{1+\alpha}\cdot\log n$, where $0<\alpha<1$ is a constant.
\end{itemize}
\end{corollary}

The works most closely related to \Cref{prob:pointwise-binary-error-lwe} are attacks based on linearization, Gr\"obner bases, and BKW
\cite{AG11,SemaevTenti21,KS99,CKPS00,FIMS03,Ivanyos07,IPS18,BKW03}.

For every $q\ge3$ and $m=\Theta(n)$, the best unconditional bounds remain
exponential, with precise exponents analyzed in~\cite{BCDL19,CDE21,KL22,WANT26,BBPSWW23,BBO24}.
For $m\ge\kappa n$, guessing all but $d=\Theta_q(\sqrt m)$ secret coordinates and applying the known attacks in the remaining dimension gives $q^{n-d}$ runtime.  Semi-regularity \emph{heuristically} predicts a much better curve $(m/n)^{O_q(\lceil n^2/m\rceil)}$~\cite{ACFP14,ACFP15,ACFFP15BKW,STA20,Steiner24}, which \Cref{thm:lwe-uniform} \emph{unconditionally} proves in the stronger pointwise sense.
For polynomial-time attacks, the prior best rigorous guarantee was $m\ge(\kappa_q/2+o(1))n^2$~\cite{ACFP14}, which \Cref{cor:lwe-uniform} improves to $m\ge\eps\cdot n^2$ for every constant $\eps>0$.

More loosely related are binary-error LWE hardness reductions, growing-modulus
BKW, hybrid lattice/meet-in-the-middle attacks, bounded-error algebraic attacks,
and cryptosystems based on nearby binary-error or restricted-decoding
assumptions~\cite{MP13,KF15,BGPW16,CMSU25,BGGOP16,GJLS21,BBWW26}.  Their
matrix structure, secret or error distributions, or sample rates differ from \Cref{prob:pointwise-binary-error-lwe}.

\subsection{\texorpdfstring{Classical algorithms for \FqSubsetSum}{Classical algorithms for Fq-Subset-Sum}}\label{sec:intro-fq}

To assess the prospects for an exponential quantum speedup for \FSubsetSum, we need to compare our quantum algorithm to the best classical algorithms.
Prior work~\cite{KOW26} gave an efficient classical algorithm that uses $m \approx n^2/3$ equations.
We improve this to $m = c_\star\cdot n^2$, where
\begin{equation}\label{eq:c-star}
c_\star=\frac{8\log(3)-9}{36}\approx0.1022.
\end{equation}
The algorithm that achieves this constant is a convoluted affair that builds on the strategy of~\cite{KOW26} , but uses several LLM-constructed gadgets that improve the performance.
We do not expect this algorithm to be the best-possible and have not found it insightful to study it in detail.
We thus do not give a complete description or formal proof of the classical algorithm in this paper and leave it to future work to develop a more systematic understanding of classical algorithms for \FSubsetSum.

Finally, we consider the subset-sum problem with a larger prime modulus.

\begin{problem}[\FqSubsetSum]\label{prob:fp-subset-sum}
Let $q\ge3$ be a prime.  
Given input vectors $h_1,\ldots,h_m\in\F_q^n$, the task is to find a nonempty $S\subseteq[m]$ such that $\sum_{i\in S}h_i=0$.  
In the average case, $h_i$'s are independently uniform; in the worst case, they are arbitrary.
\end{problem}

For the average-case \Cref{prob:fp-subset-sum}, the prior work~\cite{KOW26} gave a deterministic classical algorithm whenever $m\ge\Omega_q\left(n^{2\lfloor(q+3)/4\rfloor}\right)$. We give a near-quadratic improvement here.

\begin{restatable}{theorem}{generalpaverage}\label{thm:general-p-average}
For every
$$
m\ge q^{q^2}\cdot n^{1+\lfloor(q+3)/4\rfloor},
$$
there is a deterministic classical algorithm that solves the average-case \Cref{prob:fp-subset-sum} with probability at least $1-2^{-\Omega(n)}$ and runs in time $\poly(m)$.
\end{restatable}

\Cref{thm:general-p-average} gives a direct polynomial saving starting from $q=5$: the bound drops from $O(n^4)$~\cite{KOW26} to $O(n^3)$ in \Cref{thm:general-p-average}. In fact, our error dependence is also better: for inverse-exponential failure probability as in \Cref{thm:general-p-average}, the stated bound in~\cite{KOW26} becomes $m=\Omega(n^5)$.

For the worst-case \Cref{prob:fp-subset-sum}, previous algorithms require $n^{O(q\log q)}$ vectors~\cite{IS17,II24}. We shave the $\log q$ factor on the exponent, resulting in the following \Cref{thm:worst-main}.

\begin{restatable}{theorem}{worstmain}\label{thm:worst-main}
For every 
$$
m\ge q^{q^2}\cdot(n+1)^{q-1},
$$
there is a deterministic classical algorithm that solves the worst-case \Cref{prob:fp-subset-sum} and runs in time $\poly(m)$.
\end{restatable}

The factors $q^{q^2}$ in \Cref{thm:general-p-average,thm:worst-main} are convenient uniform upper bounds; we have not optimized their dependence on $q$.

The $q=3$ case has been discussed in \Cref{sec:intro_f3_quantum}.
For a fixed prime $q\ge5$, write $k=\lfloor(q+3)/4\rfloor$.  The best prior average-case \Cref{prob:fp-subset-sum} bound was $m\ge\Omega_q\left(n^{2k}\log(1/\eta)\right)$ for success probability $1-\eta$~\cite{KOW26}, following a quantum algorithm requiring $m\ge\Omega_q(n^{q-1})$~\cite{CLZ22}. Thus, at $\eta=2^{-\Omega(n)}$, \Cref{thm:general-p-average} lowers the $n$-exponent from $2k+1$ to $k+1$.
In the worst-case \Cref{prob:fp-subset-sum}, the best prior bound was $m\ge\Omega_q\left(n^{(\bar q/4)\log\bar q}\right)$, where $\bar q$ is the least power of two at least $q$~\cite{IS17,II24,KOW26}.
Our \Cref{thm:worst-main} replaces its $(\bar q/4)\log\bar q=\Theta(q\log q)$ exponent of $n$ by $q-1$.

For smaller $m$, the best asymptotic attacks remain Wagner-type algorithms~\cite{Wagner02}, which have runtime $2^{O_q(n/\log(m/n))}$; and specifically at linear rate $m=\Theta(n)$, the known attacks remain exponential~\cite{BBPSWW23,BBO24}.

The cryptographic hash of~\cite{IN96} over $\F_q^n$ is proposed in the linear-density range $m=\Omega(n)$; none of the polynomial-time results in this subsection reaches that range. 
More loosely related are small-parameter SIS and LWE assumptions~\cite{MP13}, Wagner analyses for
$\mathrm{SIS}^{\infty}$~\cite{DEL25}, other short-relation assumptions
\cite{KOW26}, and code-based restricted-decoding assumptions and systems
\cite{Wave19,Wave23,BBWW26}.  These problems allow signed, bounded, or arbitrary nonzero coefficients, and often use growing moduli or structured matrices.

\section*{Acknowledgments and AI disclosure}\label{sec:discussion}

KW thanks Pravesh Kothari, Shachar Lovett, and Mingjia Zhang for helpful references and discussions.

This project began in summer 2025. Before the release of GPT 5.6, the authors had developed the reduction from \Cref{prob:f3-subset-sum} to \Cref{prob:pointwise-binary-error-lwe} and formulated the XL/relinearization attack for the latter using higher-order Fourier analysis and algebraic geometry. These approaches involved connections to strength and Birch rank~\cite{baily2024strength} and to Fr\"oberg's conjecture~\cite{froberg1985inequality}, respectively. The authors had also proved several variants and related results, including an analogue of \Cref{thm:lwe-uniform} over sufficiently large finite fields. This final result was obtained with the help of the Google DeepMind AI co-mathematician~\cite{GDM26}. Aside from this, during this phase, assistance from LLMs was limited primarily to literature searches and explanations of related topics.

\noindent After its release, GPT 5.6 supplied three key ingredients missing from the earlier attempts:
\begin{enumerate}[nosep]
\item the multiplicity Schwartz--Zippel lemma, as in \Cref{lem:finite-rank-progress};
\item a multiset analogue of Lov\'asz's shadow bound, stated in \Cref{thm:multiset-shadow};
\item a sharper shifting argument underlying \Cref{thm:truncated-profile}.
\end{enumerate}
The authors had previously explored weaker versions of the second and third ingredients, using the Kruskal--Katona theorem and initial ideals in Gr\"obner-basis theory, respectively. The missing rank argument was the principal obstacle: combining it with those earlier tools already yields most of the tradeoff curve in \Cref{thm:lwe-uniform}, and hence in \Cref{thm:quantum-main}, through the XL/relinearization attack. Subsequent interactions with GPT 5.6 supplied the stronger tools above and led to the quotient decoder, which gives the full tradeoff curve.\footnote{However we conjecture that the XL/relinearization attack may work just as well. As a first step towards a rigorous understanding, in \Cref{app:relinearization-experiments}, we discuss a conjectural tight bound in a basic case and supply experiments.} The additional classical results described in \Cref{sec:intro-fq} also arose from interactions with GPT 5.6. At that stage, the constant $c_\star$ in \eqref{eq:c-star} was approximately $0.1069$.

After the release of GPT 6, further interactions aimed at improving the classical algorithm for \Cref{prob:f3-subset-sum} reduced $c_\star$ to the value in \eqref{eq:c-star}, largely through improved search capability. As noted in \Cref{sec:intro-fq}, we do not give a formal proof of this bound in the present paper. These searches have not produced a systematic way to obtain a polynomial-time classical algorithm for $m\ge c\cdot n^2$ for every fixed $c>0$. Such an improvement, if attainable, would close the gap between the classical and quantum polynomial-time sample thresholds discussed in \Cref{sec:intro_f3_quantum}.

GPT 5.6 and GPT 6 were also used to produce initial drafts, which the authors subsequently rewrote in full. The authors remain responsible for all mathematical claims and exposition.

\section{Sketch of the Binary-Error LWE analysis}

Here we give an overview of our (classical) algorithm for \PointwiseBinaryErrorLWE. 
For simplicity we focus on the $\F_3$ case, although very little change is needed to handle the general $\F_q$ case.

\subsection{Reduction to degree of regularity}
We think of the input as a sequence of~$m$ ``two-possibility equations'' over variables $x_1, \dots, x_n$, 
\begin{equation}
    b - a \cdot x \in \{0,1\},
\end{equation}
where the vectors $a \in \F_3^n$ are chosen uniformly at random, but with the promise that there is at least one solution $x = s$.  
We immediately convert these into equivalent quadratic \emph{equations},
\begin{equation}    \label{eqn:quads}
    q(x) \coloneqq (b- \ell(x))(b-\ell(x)-1)=0, \quad \ell(x) \coloneqq a \cdot x.
\end{equation}
Thus we have to solve a system of quadratic equations $q_1(x) = \cdots = q_m(x) = 0$, where the degree-$2$ part $\ell_i(x)^2$ of $q_i(x)$ is the square of a random linear form.
We also include the field equations $x_i^3 - x_i = 0$ for all $i \in [n]$, which do not change the solution set over $\F_3$.\\

It is natural to try to use a Gr\"{o}bner Basis-type algorithm to solve this system of equations.  
To analyze the running time of such an algorithm, we need to get an upper bound on the degree of all polynomials encountered by the algorithm.
By work of Semaev--Tenti~\cite{SemaevTenti21} and Salizzoni~\cite{salizzoni2023upper}, over a finite field, this can be bounded by the \emph{degree of regularity} of the system (plus one). 
\begin{definition}
    The \emph{degree of regularity} $d_{\text{reg}}(f_1, \dots f_m)$ of a collection of $n$-variate polynomials $f_1, f_2, \cdots, f_m$ is defined to be the least~$D$ such that every degree-$D$ monomial can be written as $f_1^{\text{top}} \cdot g_1 + f_2^{\text{top}} \cdot g_2 + \cdots$, where $f_i^{\text{top}}$ is the top-degree homogeneous part of $f_i$ and the $g_i$ are suitable polynomial multipliers, which may depend on the target monomial.
\end{definition}

Intuitively, a degree of regularity at most $D$ means that when working with the polynomial system of equations $f_1=\cdots=f_m=0$, we can rewrite every degree-$D$ monomial as a polynomial of smaller degree. Hence every polynomial (of any degree) can be represented, modulo the system of equations, as a polynomial of degree strictly less than $D$. This is why a small degree of regularity helps us solve the system efficiently.

In our case, we have that $q_i^{\text{top}}$ is $\ell_i^2$; thus from~\cite{SemaevTenti21,salizzoni2023upper} we get:
\begin{theorem} \label{thm:salizzoni}
    Let $2 \leq D \leq n/2$, and let $N_D$ denote the number of monomials in $n$ variables of total degree~$D$ and individual degree at most~$2$. Then $N_D = (\Theta(n/D))^D$.
    There is a $\poly(N_D, m)$-time classical algorithm that, given a solvable system of quadratic equations $q_1(x) = \cdots = q_m(x) = 0$ over $\F_3$ as in \Cref{eqn:quads}, outputs a solution whenever
    \[
        d_{\textnormal{reg}}(\ell_1^2, \dots, \ell_m^2, x_1^3, \dots, x_n^3) \leq D.
    \]
\end{theorem}

The degree-of-regularity condition means that every monomial $M$ of total degree~$D$ admits a representation
\begin{align*}
M = \sum_{i=1}^m \ell_i(x)^2 g_i(x) + \sum_{j=1}^n x_j^3 h_j(x),
\end{align*}
where the $g_i$ and $h_j$ are homogeneous of degrees $D-2$ and $D-3$, respectively. 
The cubic generators $x_j^3$ are the top-degree parts of the field equations $x_j^3-x_j=0$ that we added to our polynomial system. 
Equivalently, since the second sum already covers every monomial divisible by some $x_j^3$, we only need to express the remaining degree-$D$ monomials (those with individual degree at most~$2$) as $\sum_i \ell_i^2 g_i$ modulo $(x_1^3,\ldots,x_n^3)$, i.e., we require that 
\begin{align*}
M \equiv \sum_{i=1}^m \ell_i(x)^2 g_i(x) \pmod{(x_1^3,\ldots,x_n^3)}.
\end{align*}

\Cref{thm:salizzoni} reduces the proof of \Cref{thm:lwe-uniform} to the following degree-of-regularity bound:
\begin{theorem}
    Let $m \geq \kappa n$, where $\kappa > \log_{3/2}(3)$.  Then with probability at least $1 - \exp(-\Omega(m))$ over the \PointwiseBinaryErrorLWE instance, we have $d_{\textnormal{reg}}(\ell_1^2, \dots, \ell_m^2, x_1^3, \dots, x_n^3) \leq O(\lceil n^2/m\rceil)$.
\end{theorem}

For the remainder of this section, we will think of $m = \eps n^2$ for an arbitrarily small constant~$\eps$. 
We will then  sketch why the degree of regularity is $O(\log(1/\eps)/\eps)$ with high probability.  This already gives a polynomial-time algorithm for  fixed $\epsilon$.  

\subsection{Outline of the degree of regularity bound}

Throughout the remainder of this sketch, homogeneous polynomial calculations are performed modulo $(x_1^3,\ldots,x_n^3)$, i.e., we only care about monomials with individual degree at most 2 and will not keep mentioning this condition in this overview.

Fixing a (constant) degree bound~$D$, our goal is to show that given $m = O(n^2/D)$ random linear forms $\ell_i = a_1 x_1 + \cdots a_n x_n$, it is very likely that all monomials of degree~$D$ are ``achievable from the $\ell_i$'s''.  
By ``achievable'', here, we mean ``expressible as $\sum_{i=1}^m \ell_i^2 g_i$ for homogeneous $g_i$'s of degree $D-2$ (modulo cubic monomials)''.

Let $V_t$ denote the vector space of all degree-$D$ homogeneous polynomials that are achievable from $\ell_1, \dots, \ell_t$.  The codimension $h_t \coloneqq \codim V_t$ (within the $N_D$-dimensional space spanned by the degree-$D$ monomials with individual degree at most~$2$) represents the number of dimensions ``still to go''.
Our goal is to show that $h_{t+1}$ has a good chance of being noticeably smaller than~$h_t$.
In the final argument, we will precisely show that the probability over the choice of the $(t+1)$-th equation $\ell_{t+1}$ (conditioned on any $\ell_1, \dots, \ell_t$):
\begin{equation} \label{eqn:thebest}
    \Pr[h_{t+1} \leq h_t - c \cdot h_t^{1-2/D}] \geq c
\end{equation}
for some universal $c > 0$.
From this, it is straightforward to show that for  $m = O(n^2/D)$ we get to $h_m = 0$ except with probability at most $\exp(-\Omega(m))$.\footnote{Briefly, write $\Phi(h) = h^{2/D}$. Then on a ``successful'' step we get $\Phi(h_{t+1}) \leq (h_t - c \cdot h_t^{1-2/D})^{2/D} \leq h_t^{2/D} - c\cdot (2/D) = \Phi(h_t) - \Theta(1/D)$ by concavity.  We have $\Phi(h_0) = \Phi(N_D) = \Theta(n/D)^2$, so we need $\Theta(n^2/D)$ successful steps to finish.}
However, in this outline we only describe how to establish something slightly weaker than \Cref{eqn:thebest}.  Specifically, we show
\begin{align}
    \Pr[h_{t+1} &\leq h_t - (c D^2/n^2)\cdot h_t ] \geq c, \text{ if $h_t \gtrsim (D^2/n) N_D$;}\label{eqn:popular-pair} \\
    \Pr[h_{t+1} &\leq h_t - (c/D^2) \cdot h_t^{1-2/D}] \geq c, \text{ in general.} \label{eqn:kk}
\end{align}
For this, we get that $m = O(n^2 \log(D)/D)$ suffices to reach $h_m = 0$ except with probability $\exp(-\Omega(m))$.\footnote{Initially, $h_t$ halves every $O(n^2/D^2)$ steps.  After $\Theta(D \log D)$ halvings it gets down to $\Theta(n/D^2)^D$.  Now the $\Phi$-potential argument from before shows that $\Phi$ decreases by $\Theta(1/D^3)$ per successful steps, and $O(n^2/D)$ such steps reduce it from $\Theta(n^2/D^4)$.}

Towards proving \Cref{eqn:popular-pair,eqn:kk}, let us fix the monomial basis for all degree-$D$ homogeneous polynomials (in lexicographic order), and think of a typical such polynomial as a row vector of length~$N_D$ (where the coefficients are the entries of the vector).
A basis for $V_t$ consists of $N_D - h_t$ such row vectors, and by performing row-reduction (and then reordering monomials) we can assume these vectors have a shape like the following:
\begin{equation} \label{eqn:list}
    \begin{bmatrix}
        I & M
    \end{bmatrix},
\end{equation}
where the identity matrix block is $(N_D - h_t) \times (N_D - h_t)$, and $M$ is $(N_D - h_t) \times h_t$.  Let the monomials labeling the columns of $M$ be $x^{A_1}, \dots, x^{A_{h_t}}$.
These monomials are ``currently unachievable'', and it would be desirable to show that a large fraction of them is likely to become achievable once $\ell_{t+1}$.  
It is actually a bit too much to hope for that the bare monomials become achievable, but they will guide us in the right direction.

\subsection{First argument: a popular pair of variables}
Here we sketch why the first bound \Cref{eqn:popular-pair} holds.
Since only $O((D^2/n)N_D)$ monomials are non-multilinear, the assumed lower bound on $h_t$ ensures that a constant fraction of the remaining monomials are multilinear.
We restrict the following argument to these monomials.
By averaging, there must exist two particular variables $x_i, x_j$ that co-occur in a $\Omega(D^2/n^2)$ fraction of the monomials $x^{A_k}$.  Assume for the sake of notation that these variables are $x_1, x_2$, and that $x^{A_1}, \dots, x^{A_r}$ are the monomials containing them, with $r \gtrsim (D^2/n^2) h_t$.  Write $x^{A_k} = x_1x_2 x^{B_k}$ for sets $B_1, \dots, B_r$ of cardinality $D-2$.  

Suppose, temporarily, that instead of getting just one incoming $\ell_{t+1} = a \cdot x$ (with $a$ random), we got four new ones, with coefficient vectors $a$, $a+e_1$, $a+e_2$, $a+e_1+e_2$. Then for each $B_k$, we could newly achieve
\begin{equation}
    (a \cdot x)^2 x^{B_k} - ((a +e_1) \cdot x)^2 x^{B_k} - ((a +e_2) \cdot x)^2x^{B_k} + ((a +e_1 + e_2) \cdot x)^2 x^{B_k} = 2x_1x_2 \cdot x^{B_k} = 2 x^{A_k}.
\end{equation}
Thus we would  get $r$ newly achievable monomials, and these constitute $r$ newly achievable dimensions by virtue of \Cref{eqn:list}.  It follows that one of the four coefficient vectors must gain us at least $r/4$ newly achievable dimensions.  
But each of $a$, $a+e_1$, $a+e_2$, $a+e_1+e_2$ individually is uniformly random.  Thus we may think of the incoming $\ell_{t+1}$ as being $b \cdot x$, where first $a$ is chosen at random, and then $b$ is randomly selected to be one of $a$, $a+e_1$, $a+e_2$, $a+e_1+e_2$.  And now we see there is at least a $1/4$ chance that $\ell_{t+1}$ will gain us at least $r/4 = \Omega(D^2/n^2) h_t$ newly achievable dimensions, verifying \Cref{eqn:popular-pair}.

\subsection{Second argument: Kruskal--Katona}
We sketch the second bound \Cref{eqn:kk} in the special case where all the remaining monomials $x^{A_1}, \dots, x^{A_{h_t}}$ are multilinear. The general case, where some remaining monomials contain squared variables, requires an extension that we omit from this overview.
Introduce the set system $\mathcal{A} = \{A_1, \dots, A_{h_t}\}$, and define
\begin{equation}
    \mathcal{S} = \partial_2 \mathcal{A} = \{B : B \text{ is a size-$(D-2)$ subset of some } A_k\}.
\end{equation}
Lov\'asz's bound for the Kruskal--Katona theorem~\cite[Exercise~13.31(b)]{Lovasz93}, applied twice, implies that $|\mathcal{S}| \gtrsim h_t^{1-2/D}$. Given the new $\ell_{t+1} = a \cdot x$, let us append each $\ell_{t+1}^2 \cdot x^B$ for $B \in \mathcal{S}$ as a new row into the matrix from \Cref{eqn:list}.  We write the new matrix as 
\begin{equation}
    \begin{bmatrix}
        I & M \\
        P(a) & Q(a)
    \end{bmatrix}.
\end{equation}
To understand the new number of newly achieved dimensions, we row-reduce the above further, obtaining 
\begin{equation}
    \begin{bmatrix}
        I & M \\
        0 & R(a)
    \end{bmatrix}, \quad R(a) = Q(a) - P(a) M.
\end{equation}
The number of newly achieved dimensions is at least $\rank R(a)$.

To analyze this, let us temporarily treat $a_1, \dots, a_n$ as \emph{indeterminates}, $z_1, \dots, z_n$.  We now wish to analyze the \emph{symbolic} rank of $R(z)$.  Specifically, we will ultimately show that
\begin{equation}    \label{eqn:symrank}
    \rank_{\F_3(z_1, \dots, z_n)} R(z) \geq |\mathcal{S}|/\binom{D}{2} \gtrsim h_t^{1-2/D}/D^2.
\end{equation}
Now when random scalars $a_i$ are substituted for the $z_i$'s, the resulting numerical rank can become smaller. But a standard argument shows that there is a constant chance that a constant fraction of the rank is retained.  Specifically, if the symbolic rank of $R(z)$ is $r$, select an $r \times r$ submatrix $T(z)$ with nonzero determinant.  Since the entries of $T(z)$ are quadratic in the $z_i$'s, $p(z) \coloneqq \det T(z)$ has degree degree at most $2r$.  Moreover, it is a simple fact that the rank deficiency of $r - \rank_{\F_3} T(a)$ is at most the multiplicity of $a$ as a zero of $p(z)$. The multiplicity version of Schwartz--Zippel~\cite{DKSS13} then implies that $\E[\mathrm{mult}_a p(z)] \leq \frac{\deg p(z)}{3} \leq (2/3) r$; hence the expected rank of $T(a)$ and hence $R(a)$ is at least~$r/3 = \Omega(h_t^{1-2/D}/D^2)$.  This confirms \Cref{eqn:kk}.

Thus it remains to show the symbolic rank bound \Cref{eqn:symrank}. To do this, first ignore the correction term $P(z) M$ and focus on the rank of $Q(z)$.  Let $H$ be be $|\mathcal{S}| \times h_t$ incidence matrix between $\mathcal{S}$ and $\mathcal{A}$; i.e., $H_{B,A} = 1[B \subseteq A]$.  Then $Q(z)_{B,A}$ is $2 z_i z_j$ if $A = B \cup \{i,j\}$, and $0$ otherwise. Hence
\begin{equation}
    Q(z) = 2\diag(z^{-B})_{B \in \mathcal{S}} \cdot H \cdot \diag(z^A)_{A \in \mathcal{A}},
\end{equation}
and it follows that the symbolic rank of $Q(z)$ equals $\rank H$ (as the diagonal matrices are invertible over $\F_3(z_1, \dots, z_n)$).  We claim that $r' \coloneqq \rank H \geq |\mathcal{S}|/\binom{D}{2}$.  To see this, take $r'$ columns of $H$ spanning the column space.  Each has $\binom{D}{2}$ many $1$'s, and collectively these $1$'s must touch all $|\mathcal{S}|$ rows, as there are no all-$0$ rows.  Hence $r' \binom{D}{2} \geq |\mathcal{S}|$, as claimed.  

Thus we have shown that the symbolic rank of $Q(z)$ is at least $|\mathcal{S}|/\binom{D}{2}$, and to complete the proof we need to show that $R(z) = Q(z) - P(z)M$ has rank that is at least as large.  Let $C$ denote a generic monomial not in $\mathcal{A}$, and recall that we originally ordered monomials lexicographically. From this it follows that in the matrix $M$ (from \Cref{eqn:list}), we have that $C > A$ (in lexicographic order) implies $M_{C,A} = 0$. Thus  we may write
\begin{equation} \label{eqn:expand}
    R(z)_{B,A}  = Q(z)_{B,A}  - \sum_{C < A} P(z)_{B,C}M_{C,A},
\end{equation}
and we note that $P(z)$ is a scalar multiple of $z^{C - B}$, since it is the coefficient of $x^C$ in $(z\cdot x)^2 x^B$.

Recall that $Q(z)$ has rank $r'$ and hence has an $r' \times r'$ submatrix with nonzero determinant; for the sake of notation, say its rows are $B_1, \dots, B_{r'}$ and its columns are $A_1, \dots, A_{r'}$. Let $Q'(z)$ denote this submatrix, and let $R'(z)$ denote the
submatrix of $R(z)$ on the same rows and columns. Every nonzero term in the determinant expansion of $Q'(z)$ is a scalar multiple of the same monomial:
\begin{equation}
    \prod_{i=1}^{r'} z^{A_{\pi(i)}-B_i} = z^\delta,
    \quad \delta \coloneqq \sum_{i=1}^{r'} A_i-\sum_{i=1}^{r'} B_i,
\end{equation}
where we are identifying  sets with their exponent vectors. Thus $\det Q'(z) = c z^\delta$ for some nonzero $c \in \F_3$.

Now expand $\det R'(z)$ using \Cref{eqn:expand}. The terms that use
only the $Q(z)$ entries contribute exactly $\det Q'(z)=cz^\delta$.
Every other nonzero term uses at least one correction
$-P(z)_{B,C}M_{C,A}$. In the exponent sum, this replaces the
corresponding column label $A$ by a lexicographically smaller monomial label $C<A$. Thus for the permutation $\pi$ specifying this
term, its exponent has the form
\begin{equation}
    \sum_{i=1}^{r'} E_i-\sum_{i=1}^{r'} B_i,
\end{equation}
where each $E_i$ is either $A_{\pi(i)}$ or a monomial strictly preceding $A_{\pi(i)}$; and, at least one replacement occurs. Since lexicographic order is preserved under addition,
this exponent is strictly smaller than
\begin{equation}
    \sum_{i=1}^{r'} A_{\pi(i)}-\sum_{i=1}^{r'} B_i=\delta.
\end{equation}
As a consequence, none of the correction terms can cancel the 
$c z^\delta$ term, meaning $\det R'(z)$ is also nonzero and we have $\rank R(z) \geq \rank R'(z) = r'$ as needed.

Together with the omitted treatment of nonmultilinear monomials, these arguments establish that the degree of regularity is $O(\log(1/\eps)/\eps)$ with high probability when $m = \eps n^2$ and $\eps$ is a constant.

\section{Quantum reduction from \texorpdfstring{\FSubsetSum}{F3-Subset-Sum} to \texorpdfstring{\BinaryErrorLWE}{Binary-Error LWE}}\label{sec:quantum}

The reduction from \FSubsetSum to average-case \BinaryErrorLWE over $\F_3$ is a special
case of the filtering technique of Chen, Liu, and Zhandry~\cite{CLZ22}, which in turn
builds on Regev's quantum reduction~\cite{Reg09} and follows the Fourier-state reinterpretation \cite{SSTX09} (see also \cite{DRT24,CT23,CT24,BCT25}).

The reduction is formally captured by the following theorem.

\begin{theorem}\label{thm:coherent-reduction}
Suppose we have a deterministic decoder
$\Dec:\F_3^{n\times m}\times\F_3^m\to\F_3^n$ that runs in time $T$ and solves
average-case \BinaryErrorLWE over $\F_3$ with error $\eta$, i.e.,
\[
 \Pr_{A\sim\F_3^{n\times m},\,s\sim\F_3^n,\,e\sim\{0,1\}^m}
 [\Dec(A,A^\top s+e)\ne s]\le\eta.
\]
Then there is a quantum algorithm that solves average-case \FSubsetSum (\Cref{prob:f3-subset-sum}) with
probability at least
\[
 1-2\eta-3^n\cdot(2/3)^m
\]
and runs in time $\poly(m,n,T)$.
\end{theorem}

The purpose of this section is to motivate and explain the quantum reduction underlying this theorem. We stress that this reduction is not new to this work and follows from~\cite{CLZ22}.

We first observe that \FSubsetSum is equivalent to the following signed zero-sum problem:
\begin{problem}[\FSignedSum]\label{prob:f3-signed-sum}
Given input vectors $a_1,\ldots,a_m\in\F_3^n$ uniformly at random, the task is to find coefficients $c_1, \dots, c_m \in \{\pm 1\}$ such that $\sum_{i = 1}^m c_i a_i=0$.
\end{problem}
The only difference between \FSubsetSum and \FSignedSum is that in the former, we allow coefficients 0 or 1 (but not $-1$), whereas in the latter we allow 1 or $-1$ (but not 0).
It is easy to see that these problems reduce to each other: given a uniformly random instance $h_1, \dots, h_m$ of \FSubsetSum, we can set \[
 a_i=h_i\quad(i<m),\qquad a_m=h_m-\sum_{i<m}h_i \,.
\]
This yields another uniformly random set of vectors, and solving \FSignedSum on $\{a_1, \dots, a_m\}$ straightforwardly implies a solution of \FSubsetSum on $\{h_1, \dots, h_m\}$ (and similarly for the other direction).
For pedagogical reasons, we focus on \FSignedSum in this section and explain how it reduces to \BinaryErrorLWE.

To solve \FSignedSum for a given set of vector $\{a_1, \dots, a_m\}$, it suffices to prepare the state (omitting normalisation)
\begin{equation}\label{eq:qr-solution}
 \ket{\Psi_{a_1, \dots, a_m}}\ \propto
 \sum_{c\in\{\pm1\}^m}\indicator\left[\sum c_i a_i = 0\right]\ket c.
\end{equation}
This is the uniform superposition over all solutions to \FSignedSum, so measuring in the computational basis will yield a uniformly random solution.

We will now rewrite this state in a more suggestive way.
Let $\omega=e^{2\pi\complexi/3}$. For every $y\in\F_3^n$,
\begin{equation}\label{eq:qr-indicator}
 \indicator[y=0]=3^{-n}\sum_{s\in\F_3^n}\omega^{\langle s,y\rangle}.
\end{equation}
Applying this for the indicator function in \Cref{eq:qr-solution}, we can write 
\begin{align*}
 \ket{\Psi_{a_1, \dots, a_m}}
 &\propto \sum_{s\in\F_3^n}\sum_{c\in\{\pm1\}^m}
       \omega^{\sum_i c_i\langle a_i,s\rangle}\ket c.
\end{align*}
Expanding the sum over $c$, this can can be rewritten as 
\begin{align*}
 \ket{\Psi_{a_1, \dots, a_m}}
 &=3^{-n}\sum_{s\in\F_3^n}\bigotimes_{i=1}^m
       \left(\omega^{\langle a_i,s\rangle}\ket{+1}
             +\omega^{-\langle a_i,s\rangle}\ket{-1}\right)\\
 &=\frac{2^{m/2}}{3^n}\sum_{s\in\F_3^n}\bigotimes_{i=1}^m
       \ket{\tau_{\langle a_i,s\rangle}} \,,
\end{align*}
where we introduced the \emph{trine states} 
\begin{align*}
 \ket{\tau_a}:=\frac{\omega^a\ket{+1}+\omega^{-a}\ket{-1}}{\sqrt2},
 \qquad a\in\F_3.
\end{align*}

The trine states are normalised states that have pairwise overlap $\langle\tau_a\mid\tau_b\rangle=-1/2$ for $a\ne b$.
Consider the following alternative set of normalised states with the same pairwise overlaps: 
\begin{align*}
\ket{\text{not 0}} = \frac{\ket1-\ket2}{\sqrt2},\quad
\ket{\text{not 1}} = \frac{\ket2-\ket0}{\sqrt2},\quad
\ket{\text{not 2}} = \frac{\ket0-\ket1}{\sqrt2}.
\end{align*}
Since this set of states has the same pairwise overlaps (i.e., the same Gram matrix) as the trine states, there exists a unitary $U$ on $\mathbb{C}^3$ such that 
\begin{align*}
\ket{\text{not 0}} = U \ket{\tau_0},\quad
\ket{\text{not 1}} = U \ket{\tau_1},\quad
\ket{\text{not 2}} = U \ket{\tau_2}.
\end{align*}
This unitary is efficient to apply to each qutrit in $\ket{\Psi_{a_1, \dots, a_m}}$.
Hence, we have established that to solve the \FSignedSum problem, it suffices to prepare the state 
\begin{align*}
\ket{\tilde \Psi_{a_1, \dots, a_m}}
 &=U^{\otimes m}\ket{\Psi_{a_1, \dots, a_m}}\\
 & \propto \sum_{s\in\F_3^n}\bigotimes_{i=1}^m
       \ket{\text{not }\langle a_i,s\rangle} .
\end{align*}

How might one prepare this state?
A first observation is that given $a_i$ and $s$, it is easy to prepare $\bigotimes_{i=1}^m \ket{\text{not }\langle a_i,s\rangle}$.
This means that there exists an efficient unitary $V$ such that 
\begin{align*}
V \ket{0} \ket{s} = \bigotimes_{i=1}^m \ket{\text{not }\langle a_i,s\rangle} \ket{s} \,.
\end{align*}
Consequently, by preparing a uniform superposition over secrets $s$ and applying $V$, we can efficiently prepare a state proportional to 
\begin{align*}
\sum_{s\in\F_3^n}\bigotimes_{i=1}^m
       \ket{\text{not }\langle a_i,s\rangle} \ket{s} \,.
\end{align*}
This is almost our desired state $\ket{\tilde \Psi_{a_1, \dots, a_m}}$, except that we have an additional register with the state $\ket{s}$.
We therefore need to uncompute this additional register, i.e., we need to coherently determine $s$ from $\bigotimes_{i=1}^m \ket{\text{not }\langle a_i,s\rangle}$.

Suppose first that we measure the registers $\bigotimes_{i=1}^m \ket{\text{not }\langle a_i,s\rangle}$ in the standard basis. 
By definition of the (suggestively named) $\ket{\text{not } i}$-states, this will yield a tuple $(b_1, \dots, b_m) \in \mathbb{F}_3^m$ with the property that $\langle a_i,s\rangle \neq b_i$.
In fact, each $b_i$ will be chosen uniformly at random from the two possible $\mathbb{F}_3$-elements that are not equal to $\langle a_i,s\rangle$.
This means that $(A,b-\bone)$, where $A=(a_1|\cdots|a_m)$ and $b=(b_1,\ldots,b_m)$, is a (uniformly chosen) instance of \BinaryErrorLWE, so the decoder $\Dec$ whose existence we assumed in Theorem~\ref{thm:coherent-reduction} can recover $s$ (with high probability).

This reasoning also works in superposition; the only difficulty is handling the decoder errors.
Suppose first that we have a perfect decoder, i.e., consider the special case $\eta = 0$ in \Cref{thm:coherent-reduction}.\footnote{Note that this is an idealization: even information-theoretically, a perfect decoder is impossible because different secrets $s$ can produce the same value of $A^\top s + e$ for different (albeit highly unlikely) choices of the error vector $e$. We handle decoder errors in \Cref{sec:qr-errors}.}
Denoting the operation of the decoder by 
\begin{equation}\label{eq:qr-uncompute}
 W_A:\ket b\ket r\longmapsto
       \ket b\ket{r-\Dec(A,b-\bone)} \,,
\end{equation}
we get that 
\begin{align*}
 W_A\left(\sum_{s\in\F_3^n}
       \left(\bigotimes_{i=1}^m\ket{\text{not }\langle a_i,s\rangle}\right)
       \ket s\right)
 &=\left(\sum_{s\in\F_3^n}\bigotimes_{i=1}^m
       \ket{\text{not }\langle a_i,s\rangle}\right)\ket0\\
 &\propto\ket{\tilde\Psi_{a_1,\ldots,a_m}}\ket0
\end{align*}
as desired.

\subsection{Allowing decoding errors}\label{sec:qr-errors}

To handle decoding errors, let us first write down the full, normalised
state immediately after applying $W_A$.
Each ``not'' state can be written as the signed equal superposition
\[
 \ket{\text{not }a}
 =\frac{\ket{a+1}-\ket{a-1}}{\sqrt2}
 =\frac{1}{\sqrt2}\sum_{\delta\in\{\pm1\}}\delta\ket{a+\delta},
 \qquad a\in\F_3,
\]
where addition is in $\mathbb{F}_3$.
Writing $A=(a_1|\cdots|a_m)$ as before and introducing the normalisation factor that we previously omitted, the full quantum state after applying $W_A$ therefore is
\begin{align}
 &\frac{1}{\sqrt{3^n2^m}}
       \sum_{\substack{s\in\F_3^n\\\delta\in\{\pm1\}^m}}
       \left(\prod_{i=1}^m\delta_i\right)\ket{A^\top s+\delta} \label{eq:total_state}
       \ket{s-\Dec(A,A^\top s+\delta-\bone)}\\
 &\qquad=\ket{G_A}\ket0+\ket{E_A},
\end{align}
where we have split the sum according to whether the decoder succeeds:
\begin{align*}
 \ket{G_A}
 &:=\frac{1}{\sqrt{3^n2^m}}
       \sum_{\substack{s,\delta:\\
                 \Dec(A,A^\top s+\delta-\bone)=s}}
       \left(\prod_{i=1}^m\delta_i\right)\ket{A^\top s+\delta},\\
 \ket{E_A}
 &:=\frac{1}{\sqrt{3^n2^m}}
       \sum_{\substack{s,\delta:\\
                 \Dec(A,A^\top s+\delta-\bone)\ne s}}
       \left(\prod_{i=1}^m\delta_i\right)\ket{A^\top s+\delta}
       \ket{s-\Dec(A,A^\top s+\delta-\bone)}.
\end{align*}
Here and below, $s$ ranges over $\F_3^n$ and $\delta$ over $\{\pm1\}^m$.
The signs in the amplitudes are real numbers, while the register labels
are computed over $\F_3$. The factor $1/\sqrt{3^n2^m}$ comes from the
uniform superposition over secrets and the $m$ normalised ``not'' states.

First consider the unsuccessful component $\ket{E_A}$.
By construction, the error probability of the decoder is equal to the probability that if we measure the second register in \Cref{eq:total_state}, we get a non-zero outcome.
The latter probability is exactly the squared norm of $\ket{E_A}$.
Averaging over $A$, the decoder's error bound therefore gives
\begin{align*}
\E_A\left[\norm{\ket{E_A}}^2\right]\le\eta \,.
\end{align*}

Now consider the successful component $\ket{G_A}\ket0$. The second
register has been erased, but the first register contains only the
correctly decoded terms rather than all combinations of $s$ and $\delta$. 
To see how it differs from the desired coherent
sum, collect the missing terms in the vector
\begin{align}
 \ket{B_A}:=\frac{1}{\sqrt{3^n2^m}}
 \sum_{\substack{s,\delta:\\
                 \Dec(A,A^\top s+\delta-\bone)\ne s}}
       \left(\prod_{i=1}^m\delta_i\right)\ket{A^\top s+\delta}. \label{eq:BA_def}
\end{align}

Adding these missing terms back recovers the desired coherent sum.
Thus the full state in \Cref{eq:total_state} can be written as
\begin{align*}
 \ket{G_A}\ket0+\ket{E_A}
 &=\left(\frac{1}{\sqrt{3^n}}\sum_{s\in\F_3^n}
       \bigotimes_{i=1}^m\ket{\text{not }\langle a_i,s\rangle}\right)\ket0 -\ket{B_A}\ket0+\ket{E_A}.
\end{align*}
The first term is proportional to the desired state
$\ket{\tilde\Psi_{a_1,\ldots,a_m}}\ket0$. Applying
$(U^\dagger)^{\otimes m}$ to the first register sends this term entirely
to valid signed solutions. Hence the probability of an invalid output
is at most the squared norm of the combined error term. Moreover,
$\ket{B_A}\ket0$ and $\ket{E_A}$ are orthogonal, since the latter has a
nonzero second register. We therefore obtain
\[
 \Pr[\text{failure}\mid A]
 \le\norm{-\ket{B_A}\ket0+\ket{E_A}}^2
 =\norm{\ket{B_A}}^2+\norm{\ket{E_A}}^2
\]
Averaging over $A$ gives
\begin{equation}\label{eq:qr-failure-bound}
 \Pr[\text{failure}]\le\eta+\E_A\left[\norm{\ket{B_A}}^2\right].
\end{equation}

It remains to bound $\norm{\ket{B_A}}^2$. 
For this, observe from \Cref{eq:BA_def} that the diagonal terms in $\norm{\ket{B_A}}^2$ sum to $\norm{\ket{E_A}}^2$.  An off-diagonal pair $(s,\delta)$ and $(s',\delta')$ can contribute
only when $A^\top s+\delta=A^\top s'+\delta'$ and
$(s,\delta)\ne(s',\delta')$, which necessarily implies $s\ne s'$.  Hence, for each such
fixed pair, $\Pr_A[A^\top(s-s')=\delta'-\delta]=3^{-m}$.  In total,
\begin{align*}
\E_A\left[\norm{\ket{B_A}}^2\right]
&\le \E_A[\norm{\ket{E_A}}^2]
  +\sum_{s\ne s'}\sum_{\delta,\delta'}3^{-m}\cdot3^{-n}2^{-m}\\
&\le\eta+3^{2n}\cdot2^{2m}\cdot3^{-m}\cdot3^{-n}\cdot 2^{-m}  = \eta+3^n(2/3)^m\,.
\end{align*}
Hence, 
\begin{align*}
 \Pr[\text{failure}]\le 2\eta+3^n(2/3)^m \,,
\end{align*}
concluding the proof of \Cref{thm:coherent-reduction}.

\section{\texorpdfstring{\PointwiseBinaryErrorLWE}{Pointwise Binary-Error LWE} via quotient decoding}\label{sec:lwe}

Throughout this section, $q$ is fixed as in \Cref{prob:pointwise-binary-error-lwe} and we will prove \Cref{thm:lwe-uniform}, which proves the heuristic tradeoff predicted by Sun, Tibouchi, and Abe~\cite{STA20}.

\lweuniform*

We prove the theorem by combining exhaustive search in the linear-sample regime $m\approx n$ with the quotient decoder at higher densities $m\gg n$.

\paragraph{Section organization.}
\Cref{sec:uniform-quotient-decoder} describes the quotient decoder algebraically.
\Cref{sec:quotient-dec-alg} gives its algorithmic version.
\Cref{sec:lwe-correctness} proves \Cref{thm:lwe-uniform} and its corollary \Cref{thm:quantum-main}, assuming a technical surjectivity condition.
\Cref{sec:lwe-success} then supplies the proof of the surjectivity bound.

\subsection{The algebraic version of the quotient decoder}\label{sec:uniform-quotient-decoder}

We first describe the algebra that encodes the candidate secrets, then in \Cref{sec:quotient-dec-alg} explain how to compute it algorithmically.

Given $(A,b)\in\F_q^{n\times m}\times\F_q^m$, let $a_i$ denote column $i$ of $A$ and set
\begin{equation}\label{eq:sec:decoder_setup}
 S=\F_q[x_1,\ldots,x_n],\qquad L_i(x)=\langle a_i,x\rangle,
 \qquad f_i(x)=(L_i(x)-b_i)(L_i(x)-b_i+1).
\end{equation}
These are the annihilating equations of Arora and Ge~\cite{AG11}: $f_i(u)=0$ exactly
when $b_i-L_i(u)\in\{0,1\}$.  Thus the set of candidate secrets is
\begin{equation}\label{eq:candidate-sec}
 \calC=\{u\in\F_q^n:f_i(u)=0\text{ for all }i\in[m]\}.
\end{equation}
In particular, the planted secret $s$ from \Cref{prob:pointwise-binary-error-lwe} belongs to $\calC$ as $b=A^\top s+e$ with
$e\in\{0,1\}^m$.
Define the quotient algebra $S/J$ where
\begin{equation}\label{eq:J}
 J=(f_1,\ldots,f_m,x_1^q-x_1,\ldots,x_n^q-x_n)\subseteq S
 \quad\text{and $(\cdot)$ denotes the generated ideal.}
\end{equation}

We first describe the quotient by the field equations and consider $S/I$, where
\[
 I=(x_1^q-x_1,\ldots,x_n^q-x_n).
\]
Observe that the evaluation map $\ev\colon S/I\to\prod_{u\in\F_q^n}\F_q$ is an isomorphism, where $\ev([p]_I)=(p(u))_{u\in\F_q^n}$ and $[p]_I=p+I$ is the coset in $S/I$ represented by $p$.

We next impose the sample equations and consider $S/J$.
Define the evaluation map
$$
\ev'\colon S/J\to\prod_{u\in\calC}\F_q
\quad\text{by}\quad
\ev'([p]_J)=(p(u))_{u\in\calC},
$$
which is simply $\ev$ restricted to $\calC$ and we use $[p]_J=p+J$ to denote cosets in $S/J$. Note that $\ev'$ is well-defined since any $f\in J$ vanishes on the entire $\calC$.

\begin{claim}\label{clm:decoder-candidate-algebra}
$\ev'$ is an isomorphism between $S/J$ and $\prod_{u\in\calC}\F_q$. Moreover, $[\delta_u]_J,u\in\calC$ form a basis of $S/J$, where $\delta_u$ is the indicator polynomial of $u$.
\end{claim}
\begin{proof}
By the definition of $\ev$, multiplication in $S/I$ becomes entrywise product in $\prod_{u\in\F_q^n}\F_q$. Hence
\[
 f[\delta_u]_I=f(u)[\delta_u]_I\quad\text{in }S/I.
\]

If $u\notin\calC$, choose $i$ with $f_i(u)\ne0$ and then $[\delta_u]_I=f_i(u)^{-1}f_i[\delta_u]_I$ vanishes after quotienting by $f_i$; in other words, $[\delta_u]_J=[0]_J$.
The remaining $[\delta_u]_J,u\in\calC$ are linearly independent. To see this, assume that $\sum_{u\in\calC}c_u[\delta_u]_J=[0]_J$. Then $\sum_{u\in\calC}c_u\delta_u\in J$, which means $\sum_{u\in\calC}c_u\delta_u(v)=0$ for every $v\in\calC$ and thus each $c_u=0$.
Hence they form a linear basis of $S/J$ and $\ev'$ is an isomorphism.
\end{proof}

As a consequence of \Cref{clm:decoder-candidate-algebra}, we know how to recover the secret $s$ in \Cref{prob:pointwise-binary-error-lwe}, at least algebraically.

\begin{corollary}\label{cor:decoder-unique-candidate}
There is a unique candidate $s$ for \Cref{prob:pointwise-binary-error-lwe} if and only if $\dim(S/J)=1$.  In that case, $[1]_J$ is a basis of $S/J$, and the $j$th coordinate of $s$ is determined by $[x_j]_J=s_j[1]_J$.
\end{corollary}
\begin{proof}
By \Cref{clm:decoder-candidate-algebra}, $\dim(S/J)=|\calC|$. When $\calC=\{s\}$, the isomorphism $\ev'$ is evaluation at $s$ which sends $[1]_J$ to $1$ and $[x_j]_J$ to $s_j$.
\end{proof}

To compute $S/J$ efficiently and for simplicity, we consider the graded ring that isolates the homogeneous part of each degree. Define
\begin{equation}
 R=S/(x_1^q,\ldots,x_n^q)
 \quad\text{and}\quad
 R_d=\text{the homogeneous degree-$d$ part of $R$}.
 \label{eq:truncated-ring}
\end{equation}
Call a monomial \emph{reduced} if each variable has exponent below $q$.
Thus $R_d$ has a basis consisting of the reduced degree-$d$ monomials.

Multiplication in $R$ discards any monomial divisible by an $x_j^q$;
these are precisely the terms whose degree drops when we reduce modulo $I$.
Equivalently, $R$ is the associated graded ring of $S/I$ for the total-degree
filtration.

The degree-two homogeneous part of $f_i$ is $L_i^2$. For an integer $D\ge2$,
define
\begin{equation}
 \Phi_{m,D}:R_{D-2}^{m}\longrightarrow R_D
 \quad\text{by}\quad
 \Phi_{m,D}(Q_1,\ldots,Q_m)=\sum_{i=1}^m L_i^2Q_i.
 \label{eq:truncated-top-map}
\end{equation}
If this map is surjective, the sample equations express every reduced degree-$D$ monomial, modulo $J$, as a polynomial of smaller degree.
We will later prove in \Cref{sec:quotient-dec-alg} that the surjectivity allows us to efficiently compute $S/J$, after which \Cref{cor:decoder-unique-candidate} determines whether the candidate is unique and, if so, recovers it.

We summarize the algebraic version of the degree-$D$ quotient decoder as the following \Cref{def:quotient-decoder}.

\begin{definition}[Degree-$D$ quotient decoder, algebraically]\label{def:quotient-decoder}
On input $(A,b)$ from \Cref{prob:pointwise-binary-error-lwe} and an integer $D\ge2$, perform the following computation.
\begin{enumerate}[label=(\arabic*)]
\item\label{itm:def:quotient-decoder_1}
If $\Phi_{m,D}$ is not surjective, output the all-zero vector in $\F_q^n$ and halt.
\item\label{itm:def:quotient-decoder_2}
Compute $S/J$. If $\dim(S/J)\ne1$, output the all-zero vector and halt.
\item\label{itm:def:quotient-decoder_3}
For each $j\in[n]$, find $s_j\in\F_q$ such that
$[x_j]_J=s_j[1]_J$; and output $s=(s_1,\ldots,s_n)$.
\end{enumerate}
\end{definition}

\subsection{The algorithmic version of the quotient decoder}\label{sec:quotient-dec-alg}

We implement \Cref{def:quotient-decoder} using polynomial arithmetic and
Gaussian elimination. 
Our algorithm is similar to a related approach using finite-field Gr\"obner bases in~\cite{SemaevTenti21}.

Given the surjectivity of $\Phi_{m,D}$, we will construct matrices that multiply by each variable
and reduce the result to degree below $D$. We then compute all linear relations
needed to make these matrices act as multiplication on $S/J$, so everything becomes \emph{linear} algebra.

For each $\alpha\in\{0,\ldots,q-1\}^n$, we use $x^\alpha$ to denote the monomial $\prod_{j=1}^nx_j^{\alpha_j}$ and use $|\alpha|=\sum_j\alpha_j$ to denote its degree.
Write
\[
 \calM_{<D}=\{x^\alpha:\alpha\in\{0,\ldots,q-1\}^n,\ |\alpha|<D\},
 \quad\text{and}\quad
 V=\langle\calM_{<D}\rangle,
\]
where $\langle\cdot\rangle$ denotes the linear span. 
We regard $V$ as a space of reduced polynomials in $S$.  By \eqref{eq:truncated-ring},
it is naturally identified, as a vector space, with
$R_0\oplus R_1\oplus\cdots\oplus R_{D-1}$. We also note the following simple fact on the dimension of $V$.

\begin{fact}\label{fct:dim-V}
$\dim(V)=|\calM_{<D}|\le\min\left\{q^n,\binom{n+D-1}{D-1}\right\}$.
\end{fact}

To address \Cref{itm:def:quotient-decoder_1} of \Cref{def:quotient-decoder}, we recall the linear map $\Phi_{m,D}$ from \eqref{eq:truncated-top-map}.
Form the matrix of $\Phi_{m,D}$ in the reduced monomial bases: its columns are $L_i^2x^\alpha$ computed in $R_D$, for $i\in[m]$ and reduced monomials $x^\alpha$ of degree $|\alpha|=D-2$.
Surjectivity of $\Phi_{m,D}$ means that this matrix has full row rank, namely
$\rank(\Phi_{m,D})=\dim(R_D)$.
In this case and by Gaussian elimination, for every reduced 
monomial $x^\alpha$ of degree $|\alpha|=D$, we have polynomials $Q_{i,\alpha}\in R_{D-2}$ satisfying
\begin{equation}\label{eq:decoder-lifted-relations-t}
 \sum_{i=1}^m L_i^2Q_{i,\alpha}=x^\alpha\quad\text{in }R_D.
\end{equation}

Now recall the actual vanishing polynomial is $f_i(x)=(L_i(x)-b_i)(L_i(x)-b_i+1)$ from \Cref{prob:pointwise-binary-error-lwe}. Hence the monomial $x^\alpha$ reduces to lower-degree terms modulo $J=(f_1,\ldots,f_m,x_1^q-x_1,\ldots,x_n^q-x_n)$.
This inspires us to define the actual remainder $r_\alpha\in V$ by
\begin{equation}
 r_\alpha=x^\alpha-\red\!\left(\sum_{i=1}^m f_iQ_{i,\alpha}\right),
 \label{eq:decoder-lifted-relations}
\end{equation}
where $\red(\cdot)$ computes the remainder modulo $I$ by repeatedly replacing $x_j^q$ by $x_j$.

\begin{claim}\label{clm:decoder-lifted-relations}
If $\Phi_{m,D}$ is surjective, then for each reduced degree-$D$ monomial $x^\alpha$, we have $r_\alpha\in V$ and $[x^\alpha]_J=[r_\alpha]_J$. Moreover, the surjectivity and these $r_\alpha$'s can be computed in deterministic time $\poly(\dim(V),n,m,\log(q))$.
\end{claim}
\begin{proof}
$r_\alpha\in V$ follows from \eqref{eq:decoder-lifted-relations-t} and $[x^\alpha]_J=[r_\alpha]_J$ is due to \eqref{eq:decoder-lifted-relations} and the definition of $\red(\cdot)$.
For the runtime, we simply note that the matrix of $\Phi_{m,D}$ has $m\dim(R_{D-2})\le m\dim(V)$ columns and $\dim(R_D)\le n\dim(V)$ rows, on which we perform $\F_q$ arithmetic and Gaussian elimination.
\end{proof}

Let $\tau$ be the linear map from reduced polynomials of degree at most $D$ to $V$
that fixes $V$ and replaces each degree-$D$ monomial $x^\alpha$ by $r_\alpha$. That is,
\begin{equation}\label{eq:tau}
\tau(x^\alpha)=\begin{cases}
    x^\alpha & |\alpha|<D,\\
    r_\alpha & |\alpha|=D,
\end{cases}
\end{equation}
which is then extended linearly to all reduced polynomials of degree at most $D$.

To address \Cref{itm:def:quotient-decoder_2} of \Cref{def:quotient-decoder}, we need to construct linear maps that capture $S/J$ upon surjectivity of $\Phi_{m,D}$.
Recall that $\red(\cdot)$ computes the reduced form of a polynomial modulo $I$.
Define linear maps $\Xsf_1,\ldots,\Xsf_n\colon V\to V$ by
\begin{equation}
 \Xsf_j(v)=\tau(\red(x_jv))\quad\text{for every $v\in V$.}
 \label{eq:decoder-multiplication-matrices}
\end{equation}
These maps satisfy $[\Xsf_j(v)]_J=[x_jv]_J$ and represent the quotient rule after multiplying by each $x_j$.

\begin{claim}\label{clm:Xsf+J}
If $\Phi_{m,D}$ is surjective, then $[\Xsf_j(v)]_J=[x_jv]_J$ holds for every $j\in[n]$ and $v\in V$. Moreover, these $\Xsf_j$'s can be computed in deterministic time $\poly(\dim(V),n,m,\log(q))$.
\end{claim}
\begin{proof}
By \eqref{eq:tau} and \Cref{clm:decoder-lifted-relations}, $[\Xsf_j(v)]_J=[\red(x_jv)]_J=[x_jv]_J$. The runtime follows immediately from the linearity and definition.
\end{proof}

Crucially these $\Xsf_1,\ldots,\Xsf_n$ need not commute on $V$ as different sequences of valid reductions can produce different representatives of the same coset due to the particular choices in \eqref{eq:decoder-lifted-relations-t}.
We therefore compute a relation space
$N\subseteq V$ such that the induced maps on $V/N$ commute and satisfy all field and sample equations.

Let $\Isf$ denote the identity on $V$ and recall \eqref{eq:sec:decoder_setup}.
Define for each $i\in[m]$ a linear map $\Fsf_i\colon V\to V$ by
\begin{equation}\label{eq:F_i}
 \Fsf_i(v)=\left((\Lsf_i-b_i\cdot\Isf)(\Lsf_i-b_i\cdot\Isf+\Isf)\right)(v)
 \quad\text{where}\quad
 \Lsf_i(v)=\left(\sum_{j=1}^n(a_i)_j\Xsf_j\right)(v),
\end{equation}
which is simply the vanishing equation $f_i(x)=(L_i(x)-b_i)(L_i(x)-b_i+1)$ with $L_i(x)=\sum_j(a_i)_jx_j$, modulo the reduction rules using $\tau$ in \eqref{eq:tau}.
Define for each $1\le j<k\le n$ a linear map $\Csf_{j,k}\colon V\to V$
\begin{equation}\label{eq:C_i}
 \Csf_{j,k}(v)=\left(\Xsf_j\Xsf_k-\Xsf_k\Xsf_j\right)(v),
\end{equation}
which is simply the commutative rule $x_jx_k=x_kx_j$, modulo the reduction rules using $\tau$ in \eqref{eq:tau}.
Define for each $j\in[n]$ a linear map $\Hsf_j\colon V\to V$
\begin{equation}\label{eq:H_i}
 \Hsf_j(v)=\left(\Xsf_j^q-\Xsf_j\right)(v),
\end{equation}
which is simply the vanishing field equation $x_j^q=x_j$, modulo the reduction rules using $\tau$ in \eqref{eq:tau}.

\begin{claim}\label{clm:quotient-alg-closure}
If $\Phi_{m,D}$ is surjective, then all the maps
$\Fsf_i$'s, $\Csf_{j,k}$'s, and
$\Hsf_j$'s have image in $J\cap V$.
\end{claim}
\begin{proof}
This follows directly from \Cref{clm:Xsf+J} and the definition of $J$ in \eqref{eq:J}.
\end{proof}

By \Cref{clm:quotient-alg-closure}, the images of these maps consist of valid
relations modulo $J$.  We iteratively take their linear span and then close it under the
maps $\Xsf_j$.  More precisely, define
\begin{align}
 N_0&=\sum_{1\le j<k\le n}\Csf_{j,k}(V)+\sum_{j=1}^n\Hsf_j(V)+\sum_{i=1}^m\Fsf_i(V),
       \label{eq:decoder-initial-relations}\\
 N_{r+1}&=N_r+\sum_{j=1}^n\Xsf_j(N_r)
 \quad\text{for}\quad r\ge0.
       \label{eq:decoder-relation-closure}
\end{align}
The three sums in $N_0$ impose commutativity, the field equations, and the sample
equations, respectively, on every vector of $V$.
Stop when $N_{r+1}=N_r$ and call the resulting space $N$. Then we have the following observation.

\begin{claim}\label{clm:N-sequence}
If $\Phi_{m,D}$ is surjective, $N\subseteq J\cap V$ and $\Xsf_j(N)\subseteq N$ holds for every $j\in[n]$. Moreover, a linear basis of $N$ can be computed in deterministic time $\poly(\dim(V),n,m,\log(q))$.
\end{claim}
\begin{proof}
The first line follows from \Cref{clm:Xsf+J} and \Cref{clm:quotient-alg-closure}. The runtime can be argued as follows: to compute $N_0$, we go over each monomial in $\calM_{<D}$, apply the corresponding $\Csf_{j,k},\Hsf_j,\Fsf_i$, and use Gaussian elimination to select a basis of the outcomes.
For each subsequent $N_{r+1}$, we only need to check if applying $\Xsf_j$ expands the basis; then adjoin if necessary and stop if no new basis is discovered. Since each strict enlargement increases the dimension, there are at most $\dim(V)$ such enlargements. The intermediate calculation has the runtime prescribed in \Cref{clm:Xsf+J} and \Cref{clm:quotient-alg-closure}.
\end{proof} 

In fact, $N$ provides an isomorphism between $V/N$ and $S/J$; the former one is something we can algorithmically compute and the latter one is what we need in \Cref{def:quotient-decoder}.

\begin{proposition}\label{prop:truncated-exact-quotient}
If $\Phi_{m,D}$ is surjective, then $V/N\cong S/J$. More precisely, define $\pi:V\to S/J$ by $\pi(v)=[v]_J$ and then $\pi$ is surjective with $\ker(\pi)=N$.
\end{proposition}
\begin{proof}
The surjectivity of $\pi$ is obvious and we focus on the kernel.
By \Cref{clm:N-sequence}, $N\subseteq\ker(\pi)$ and hence $\pi$ induces a surjective map $\overline\pi:V/N\to S/J$.

It remains to prove that $\overline\pi$ is injective, for which we construct its inverse.
Write $[v]_N=v+N$ and let $\overline\Xsf_j\colon V/N\to V/N$ be the map induced by $\Xsf_j$.  These maps are well-defined as $\Xsf_j(N)\subseteq N$ from \Cref{clm:N-sequence}.
Since $N_0\subseteq N$, the commutators, field equations, and sample equations all vanish.  In particular, $\overline\Xsf_j$'s commute, so we can properly define the linear map
\[
 \psi:S/J\longrightarrow V/N,
 \quad\text{by}\quad
 \psi([p]_J)=p(\overline\Xsf_1,\ldots,\overline\Xsf_n)[1]_N.
\]
We will show $\psi\circ\bar\pi$ is the identity on $V/N$, certifying the injectivity of $\bar\pi$.
To see this, for each monomial $x^\alpha\in\calM_{<D}$ we have $\overline\Xsf_1^{\alpha_1}\cdots\overline\Xsf_n^{\alpha_n}[1]_N=[x^\alpha]_N$, which span $V/N$. Hence $\psi(\bar\pi([p]_N))=\psi([p]_J)=p(\overline\Xsf_1,\ldots,\overline\Xsf_n)[1]_N=[p]_N$, proving that $\psi\circ\overline\pi$ is the identity. This completes the whole proof.
\end{proof}

In light of \Cref{prop:truncated-exact-quotient}, we can now address \Cref{itm:def:quotient-decoder_3} of \Cref{def:quotient-decoder}.

\begin{claim}\label{clm:alg-quotient-3}
If $\Phi_{m,D}$ is surjective and $\dim(V)=\dim(N)+1$, there is a unique candidate $s$ for \Cref{prob:pointwise-binary-error-lwe} and the $j$th coordinate of $s$ is determined by $x_j-s_j\in N$ for a unique $s_j\in\F_q$. Moreover, checking the conditions and solving for $s$ can be done in deterministic time $\poly(\dim(V),n,m,\log(q))$.
\end{claim}
\begin{proof}
The uniqueness follows from \Cref{prop:truncated-exact-quotient} and \Cref{cor:decoder-unique-candidate}. In that case, $[x_j]_J=s_j[1]_J$ is equivalent to $[x_j]_N=s_j[1]_N$, or $x_j-s_j\in N$. The runtime of checking the conditions follows from \Cref{clm:decoder-lifted-relations} and \Cref{clm:N-sequence}. To recover $s_j$, let $v_1,\ldots,v_r$ be the computed basis of $N$ and solve $x_j=\sum_{\ell=1}^r c_\ell v_\ell+s_j\cdot1$ by Gaussian elimination. Since $[1]_N$ is a basis of $V/N$, the family $v_1,\ldots,v_r,1$ is a basis of $V$, so this system has a unique solution and the prescribed runtime; no enumeration over $\F_q$ is needed.
\end{proof}

To summarize, we present the algorithmic version of the degree-$D$ quotient decoder, which should be compared with its algebraic version in \Cref{def:quotient-decoder}.

\begin{definition}[Degree-$D$ quotient decoder, algorithmically]\label{def:quotient-decoder-alg}
On input $(A,b)$ from \Cref{prob:pointwise-binary-error-lwe} and an integer $D\ge2$, perform the following computation.
\begin{enumerate}[label=(\alph*)]
\item\label{itm:def:quotient-decoder-alg_1}
Use \Cref{clm:decoder-lifted-relations} to check the surjectivity of $\Phi_{m,D}$.
If it is not surjective, output the all-zero vector in $\F_q^n$ and halt.
\item\label{itm:def:quotient-decoder-alg_2}
Use \Cref{clm:N-sequence} to compute $N$, which, by \Cref{prop:truncated-exact-quotient}, provides $V/N\cong S/J$.
Then use \Cref{clm:alg-quotient-3} to check if $\dim(S/J)=\dim(V/N)=1$. If not, output the all-zero vector and halt.
\item\label{itm:def:quotient-decoder-alg_3}
Use \Cref{clm:alg-quotient-3} to output the desired secret $s$.
\end{enumerate}
\end{definition}

\subsection{\texorpdfstring{Proofs of \Cref{thm:lwe-uniform} and \Cref{thm:quantum-main}}
{Proofs of the classical and quantum theorems}}
\label{sec:lwe-correctness}

To prove \Cref{thm:lwe-uniform}, we combine the implementation in
\Cref{sec:quotient-dec-alg} with the following surjectivity theorem.

\begin{restatable}{theorem}{uniformtruncatedtop}\label{thm:uniform-truncated-top}
Assume $500n\le m\le n^2$ and $D=\lceil500n^2/m\rceil$.
For a random instance from \Cref{prob:pointwise-binary-error-lwe}, we have
\[
 \Pr[\Phi_{m,D}\text{ is not surjective}]
 \le e^{-m/400}.
\]
\end{restatable}

The proof of \Cref{thm:uniform-truncated-top} is deferred to \Cref{sec:lwe-success}. We now restate and prove \Cref{thm:lwe-uniform}.

\lweuniform*

\begin{proof}
By \Cref{clm:decoder-lifted-relations}, \Cref{clm:N-sequence}, and
\Cref{clm:alg-quotient-3},
the decoder in \Cref{def:quotient-decoder-alg} runs on every input in time
$\poly(\dim(V),n,m,\log(q))$, where \Cref{fct:dim-V} upper bounds $\dim(V)$.
By \Cref{prop:truncated-exact-quotient} and \Cref{cor:decoder-unique-candidate},
it recovers $s$ whenever $\Phi_{m,D}$ is surjective and $\calC=\{s\}$.

With $C=500$, the complete algorithm in \Cref{thm:lwe-uniform} is as follows.
\begin{enumerate}[label=(\roman*)]
\item\label{itm:thm:lwe-uniform_1}
If $n<C$ or $m<Cn$, enumerate all $u\in\F_q^n$ and test membership in
$\calC$ using \eqref{eq:candidate-sec}.  Return the unique candidate if there
is one, and the all-zero vector otherwise.
\item\label{itm:thm:lwe-uniform_2}
If $n\ge C$ and $Cn\le m\le n^2$, run
\Cref{def:quotient-decoder-alg} with $D=\lceil500n^2/m\rceil$.
\item\label{itm:thm:lwe-uniform_3}
If $n\ge C$ and $m>n^2$, run the same decoder with $D=500$.
\end{enumerate}

\paragraph{Runtime.}
Put $g=m/n\ge\kappa>1$ and $h=\lceil n/g\rceil$.
Since $n\le gh$ and $h\le2^h$, every polynomial factor in $n,m$ is
$g^{O_\kappa(h)}$. It therefore suffices to bound the search space in
\Cref{itm:thm:lwe-uniform_1} and $\dim(V)$ in \Cref{itm:thm:lwe-uniform_2} and \Cref{itm:thm:lwe-uniform_3}.
In \Cref{itm:thm:lwe-uniform_1}, either $n<C$ or $g<C$, so $n\le Ch$ and
$q^n=g^{O_{q,\kappa}(h)}$.
In \Cref{itm:thm:lwe-uniform_2}, $D=\lceil500n/g\rceil\le500h$.  Since $n/D\le g/500$,
\[
 \dim(V)\le\binom{n+D}{D}
 \le\left(\frac{e(n+D)}D\right)^D
 \le\bigl(e(1+g/500)\bigr)^{500h}=g^{O(h)}.
\]
In \Cref{itm:thm:lwe-uniform_3}, $D=500$ and $\dim(V)=O(n^{499})$.  Hence all three cases give the claimed bound.

\paragraph{Failure probability.}
Fix the secret $s$ and error $e$ as in \Cref{prob:pointwise-binary-error-lwe}.
By \Cref{fct:lwe-existence},
we have $\Pr[\calC\ne\{s\}]\le2^{-\Omega_{q,\kappa}(m)}$.
Uniqueness suffices in \Cref{itm:thm:lwe-uniform_1} and the only additional failure event in \Cref{itm:thm:lwe-uniform_2} and \Cref{itm:thm:lwe-uniform_3} is that
$\Phi_{m,D}$ is not surjective.

In \Cref{itm:thm:lwe-uniform_2}, this event has probability at most $e^{-m/400}$ by \Cref{thm:uniform-truncated-top}.
In \Cref{itm:thm:lwe-uniform_3}, take $r=\lfloor m/n^2\rfloor$ disjoint blocks of $n^2$ columns.
\Cref{thm:uniform-truncated-top} bounds each block's surjectivity failure by $e^{-n^2/400}$, which implies
\[
 \Pr[\Phi_{m,500}\text{ is not surjective}]
 \le(e^{-n^2/400})^r\le e^{-m/800}.
\]
A union bound gives the claimed total failure probability in every case.
\end{proof}

The classical decoder with $q=3$ now supplies the hypothesis of
\Cref{thm:coherent-reduction}.  We use this reduction to prove \Cref{thm:quantum-main}, which we restate below.

\quantummain*

\begin{proof}
Put $g=m/n$ and $h=\lceil n/g\rceil$.
For $\kappa n\le m<3n$, exhaustive search over nonempty subsets fails
only when no solution exists.  By \Cref{fct:f3-subset-sum-existence}, its failure probability is
$2^{-\Omega_\kappa(n)}=2^{-\Omega_\kappa(m)}$.
Its runtime is $2^m\poly(m,n)=2^{O(n)}=g^{O_\kappa(h)}$
since $\kappa\le g<3$ and $n\le 3h$.

For every $m\ge 3n$, apply \Cref{thm:lwe-uniform} with $q=3$.  Its decoder runs on every input in time $g^{O(h)}$,
and its stronger pointwise error bound implies average error $2^{-\Omega(m)}$ for \Cref{thm:coherent-reduction}.
Also,
\[
 3^n(2/3)^m
 \le\exp\left[-\left(\ln(3/2)-\frac{\ln(3)}{3}\right)m\right]
 =2^{-\Omega(m)}.
\]
Thus \Cref{thm:coherent-reduction} gives failure probability $2^{-\Omega(m)}$
and runtime $\poly(n,m,g^{O(h)})=g^{O(h)}$.
\end{proof}

\subsection{Proof of the probabilistic surjectivity condition}\label{sec:lwe-success}

The remaining piece of this section is to prove \Cref{thm:uniform-truncated-top}, the probabilistic ingredient of a surjectivity condition used in \Cref{sec:lwe-correctness}.

\uniformtruncatedtop*

We recall the definition of $\Phi_{m,D}$ from \eqref{eq:truncated-top-map} for convenience:
$$
 \Phi_{m,D}:R_{D-2}^{m}\longrightarrow R_D
 \quad\text{by}\quad
 \Phi_{m,D}(Q_1,\ldots,Q_m)=\sum_{i=1}^m L_i^2Q_i,
$$
where each $L_i\in R_1$ is a uniformly random linear function.
The idea of the proof is to track the growth of the image as $m$ increases and use a coupon-collector-type martingale argument to show that each time the remaining codimension shrinks by an appropriate amount until surjectivity.

For $0\le t\le m$, let $V_t$ be the image generated by the first $t$ samples, and let $h_t$ count
the dimensions still missing from that image:
\begin{equation}\label{eq:Vt_ht}
 V_t=\sum_{i=1}^tL_i^2R_{D-2}
 \quad\text{and}\quad
 h_t=\dim(R_D)-\dim(V_t).
\end{equation}
Thus $V_0=\{0\}$, $h_0=\dim(R_D)$, and $\Phi_{m,D}$ is surjective exactly when $h_m=0$.
The number of dimensions contributed by sample $t$ is
\begin{equation}\label{eq:Xt}
 X_t=\dim(V_t)-\dim(V_{t-1})=h_{t-1}-h_t.
\end{equation}
Let $\calF_t$ denote the history of the first $t$ samples. The following \Cref{prop:conditional-image-growth}
gives an expected lower bound on $X_t$ conditioned on any such history, which gives the exponential tail by standard martingale concentration.

\begin{restatable}{proposition}{propconditionalimagegrowth}\label{prop:conditional-image-growth}
Assume $3\le D\le n$ and let $H=\dim(R_D)$.
There is a nondecreasing function $\gamma\colon[0,H]\to[0,\infty)$ such that
$$
 \E\!\left[\min\!\left\{\frac{X_t}{\gamma(h_{t-1})},1\right\}
       \,\middle|\,\calF_{t-1}\right]\ge\frac13
\quad\text{holds for every $h_{t-1}>0$,}
$$
and
$$
 \int_0^H\frac{du}{\gamma(u)}
 \le\frac{100n^2}{D}.
$$
\end{restatable}

We first deduce \Cref{thm:uniform-truncated-top}, and then prove \Cref{prop:conditional-image-growth}.

\begin{proof}[Proof of \Cref{thm:uniform-truncated-top}]
The assumptions $500n\le m\le n^2$ and $D=\lceil500n^2/m\rceil$ imply $500\le D\le n$.
For each $t\in[m]$, define a random variable $Y_t\in[0,1]$ by
\begin{equation}\label{eq:thm:uniform-truncated-top_1}
Y_t=\begin{cases}
\min\left\{\frac{X_t}{\gamma(h_{t-1})},1\right\} & h_{t-1}>0,\\
1 & h_{t-1}=0.
\end{cases}
\end{equation}
Then by \Cref{prop:conditional-image-growth} and for every $t\in[m]$, we have
\begin{equation}\label{eq:thm:uniform-truncated-top_2}
\E[Y_t\mid\calF_{t-1}]\ge1/3.
\end{equation}

Recall that the surjectivity of $\Phi_{m,D}$ is equivalent to $h_m=0$.
Hence on the failure event $h_m>0$, we know $h_t>0$ for every $t\in[m]$, which means $Y_t$ in \eqref{eq:thm:uniform-truncated-top_1} always takes the first branch and thus
\begin{align*}
\sum_{t=1}^mY_t
&\le\sum_{t=1}^m\frac{X_t}{\gamma(h_{t-1})}=\sum_{t=1}^m\frac{h_{t-1}-h_t}{\gamma(h_{t-1})}
\tag{assuming $\Phi_{m,D}$ is not surjective}\\
&\le\sum_{t=1}^m\int_{h_t}^{h_{t-1}}\frac{du}{\gamma(u)}
\tag{since $h_{t-1}\ge h_t$ and $\gamma$ is nondecreasing}\\
&\le\int_0^H\frac{du}{\gamma(u)}
\tag{since $h_0=H$ and $h_m\ge0$}\\
&\le\frac{100n^2}{D}
\tag{by \Cref{prop:conditional-image-growth}}\\
&\le\frac{100m}{500}<\frac m4=m\cdot\left(\frac13-\frac1{12}\right).
\tag{since $D=\lceil500n^2/m\rceil$}
\end{align*}
Hence by \eqref{eq:thm:uniform-truncated-top_2}, we have
$$
\frac1m\sum_{t=1}^m\left(Y_t-\E[Y_t\mid\calF_{t-1}]\right)<\frac{-1}{12}.
$$
Finally recall \eqref{eq:thm:uniform-truncated-top_1} that $Y_t\in[0,1]$. By the Azuma-Hoeffding inequality, we have
$$
 \Pr[\Phi_{m,D}\text{ is not surjective}]
 =\Pr[h_m>0]\le\Pr\left[\frac1m\sum_{t=1}^m\left(Y_t-\E[Y_t\mid\calF_{t-1}]\right)<\frac{-1}{12}\right]\le e^{-m/400}
 $$
as claimed.
\end{proof}

The proof of \Cref{prop:conditional-image-growth} has three steps: relate $X_t$ to a symbolic rank in \Cref{sec:lwe-progress-rank}, reduce symbolic rank into shadows in a multiset system in \Cref{sec:lwe-progress-shifting}, and bound the shadows in \Cref{sec:lwe-progress-shadows}.

\subsubsection{From martingale step to symbolic rank}\label{sec:lwe-progress-rank}

Recall that $X_t$ measures the dimension of new images contributed by $L_t$. This is better formulated in the dual space.

Let $R_D^*$ be the dual space of $R_D$, i.e., $R_D^*$ contains all linear functionals on $R_D$.
Define $W_t=V_t^\bot\subseteq R_D^*$ as the annihilator of $V_t$, i.e., $W_t$ contains linear functionals vanishing on the entire $V_t$.
Then $\dim(W_{t-1})=h_{t-1}$ is the number of missing dimensions given the first $t-1$ samples.

Given the $t$th sample $L_t$, it contributes $L_t^2R_{D-2}$ to the existing $V_{t-1}$ and forms $V_t=V_{t-1}+L_t^2R_{D-2}$. This inspires us to define a contraction map for every linear function $L$ as
\[
 C_L:R_D^*\longrightarrow R_{D-2}^*
 \quad\text{by}\quad
 C_L(\lambda)(Q)=\lambda(L^2Q),
\]
which is the dual of multiplication by $L^2$.

\begin{claim}\label{clm:contraction_CL}
Let $C_{L_t}|_{W_{t-1}}\colon W_{t-1}\to R_{D-2}^*$ be the restriction of $C_{L_t}$ to $W_{t-1}$.
Then
\[
 W_t=\ker(C_{L_t}|_{W_{t-1}})
 \quad\text{and}\quad
 X_t=\rank(C_{L_t}|_{W_{t-1}}).
\]
\end{claim}
\begin{proof}
Since $V_t=V_{t-1}+L_t^2R_{D-2}$, a functional $\lambda\in R_D^*$ lies in $W_t$ if and only if $\lambda\in W_{t-1}$ and $\lambda(L_t^2Q)=0$ for every $Q\in R_{D-2}$. By the definition of $C_{L_t}$, this gives $W_t=W_{t-1}\cap\ker(C_{L_t})=\ker(C_{L_t}|_{W_{t-1}})$.
Hence $\rank(C_{L_t}|_{W_{t-1}})
 =\dim(W_{t-1})-\dim(W_t)
 =h_{t-1}-h_t=X_t$.
\end{proof}

By \Cref{clm:contraction_CL}, one can define symbolic rank
\begin{equation}\label{eq:rhoW}
 \rho(W_{t-1})=\rank_{\F_q(z_1,\ldots,z_n)}(C_{L_z}|_{W_{t-1}})
 \quad\text{where}\quad
 L_z=\sum_{i=1}^n z_ix_i,
\end{equation}
where $\F_q(z_1,\ldots,z_n)$ is the field of rational functions with base field $\F_q$ and variables $z_1,\ldots,z_n$.
The point of this definition is that $C_{L_z}|_{W_{t-1}}$ specializes to $C_{L_t}|_{W_{t-1}}$ when we fix $L_z=L_t$.
While $\rho(W_{t-1})$ \emph{upper bounds} $X_t=\rank(C_{L_t}|_{W_{t-1}})$ in the worst case and we need to \emph{lower bound} the latter, we can indeed show that \emph{on average} $X_t$ is not too far from the former.

\begin{lemma}\label{lem:finite-rank-progress}
For uniformly random $L_t$, we have $\E[X_t\mid W_{t-1}]\ge\rho(W_{t-1})/3$.
\end{lemma}
\begin{proof}
The definition of $\rho:=\rho(W_{t-1})$ implies that the evaluation matrix $C_{L_z}|_{W_{t-1}}$ contains a $\rho\times\rho$ submatrix $B_z$ with nonzero determinant $\det(B_z)$. By the definition of $C_{L_z}$, each entry of $B_z$ has degree at most $2$ in $z$ and hence $\det(B_z)$ is a polynomial of degree at most $2\rho$.

For an $n$-variate nonzero polynomial $p(z)$ over the base field $\F_q$ and $a\in\F_q^n$, we use $\mult_a(p(z))$ to denote the multiplicity of the root $z=a$. Formally, $\mult_a(p(z))$ is the least total degree of a nonzero monomial in $p(z+a)$. The multiplicity Schwartz-Zippel lemma \cite{DKSS13} gives $\E_a\left[\mult_a(p(z))\right]\le\deg(p)/q$.
Applying this to $p(z)=\det(B_z)$ gives
\begin{equation}\label{eq:lem:finite-rank-progress_1}
\E_a\left[\mult_a(\det(B_z))\right]\le2\rho/q\le2\rho/3.
\end{equation}
On the other hand at any $a\in\F_q^n$, we can choose an invertible matrix $U_a$ such that $\rho-\rank(B_a)$ rows of $U_aB_a$ are zero. Hence every monomial of $\det(U_aB_{a+z})$ vanishes to at least this order, which means
\begin{equation}\label{eq:lem:finite-rank-progress_2}
 \rho-\rank(B_a)\le\mult_a(\det(U_aB_z))=\mult_a(\det(B_z)).
\end{equation}
Finally by \Cref{clm:contraction_CL}, we have 
\begin{align*}
\E[X_t\mid W_{t-1}]
&=\E\left[\rank(C_{L_t}|_{W_{t-1}})\mid W_{t-1}\right]=\E_a\left[\rank(C_{L_a}|_{W_{t-1}})\mid W_{t-1}\right]\\
&\ge\E_a\left[\rank(B_a)\mid W_{t-1}\right]
\ge\rho/3,
\tag{by \eqref{eq:lem:finite-rank-progress_1} and \eqref{eq:lem:finite-rank-progress_2}}
\end{align*}
which completes the proof.
\end{proof}

\subsubsection{From symbolic rank to multiset shadow}\label{sec:lwe-progress-shifting}

Given \Cref{lem:finite-rank-progress}, our next goal is to understand $\rho(W_{t-1})$ in \eqref{eq:rhoW}, for which, to use tools from algebraic geometry, we need to work over the algebraic closure $\overline{\F_q}$ of $\F_q$.

Define $\overline{R}_D$ as the homogeneous degree-$D$ part of $\overline{R}:=\overline{\F_q}[x_1,\ldots,x_n]/(x_1^q,\ldots,x_n^q)$.
Let $L_z=\sum_{i=1}^nz_ix_i$ and define $\overline{C}_{L_z}$ by
$$
\overline{C}_{L_z}(\lambda)(Q)=\lambda(L_z^2Q)
$$
with scalars extended to $\overline{\F_q}(z_1,\ldots,z_n)$.
Then for every $W\subseteq\overline R_D^*$, define
\begin{equation}\label{eq:rhoW+}
\bar\rho(W)=\rank_{\overline{\F_q}(z_1,\ldots,z_n)}(\overline{C}_{L_z}|_W).
\end{equation}
We note that the field extension does not change the rank.
\begin{fact}\label{fct:rhorho+}
Let $\overline W_{t-1}\subseteq\overline R_D^*$ be the span of $W_{t-1}$ over the base field $\overline{\F_q}$.
Then $\dim(\overline W_{t-1})=\dim(W_{t-1})$ and $\rho(W_{t-1})=\bar\rho(\overline W_{t-1})$.
\end{fact}
\begin{proof}
A basis of $W_{t-1}$ over $\F_q$ becomes a basis of $\overline W_{t-1}$ over $\overline{\F_q}$, giving the dimension equality. In these bases, $C_{L_z}|_{W_{t-1}}$ and $\overline C_{L_z}|_{\overline W_{t-1}}$ have the same matrix. Its minors are unchanged under the field inclusion $\F_q(z_1,\ldots,z_n)\subseteq\overline{\F_q}(z_1,\ldots,z_n)$, so its rank is unchanged.
\end{proof}

In general a subspace $W\subseteq\overline R_D^*$ is spanned by \emph{polynomial coefficient functionals}, i.e., linear functionals that assign coefficients to multiple monomials.
In the following \Cref{lem:truncated-monomial-replacement}, we will show that, in terms of lower bounding $\bar\rho(W)$, one can assume $W$ is spanned by \emph{monomial coefficient functionals}, i.e., linear functionals that assign coefficients to a single monomial.
This additional structural property will be helpful later.

Let $\calB$ be the group of invertible triangular changes of variables: every $g\in\calB$ is defined by
\[
 g(x_j)=a_jx_j+\sum_{i>j}a_{ij}x_i,
 \quad\text{where}\quad a_j\in\overline{\F_q}\setminus\{0\}
 \quad\text{and}\quad
 a_{ij}\in\overline{\F_q}.
\]
These substitutions preserve the ideal $(x_1^q,\ldots,x_n^q)$, since in characteristic $q$ we have $(\sum_i b_ix_i)^q=\sum_i b_i^qx_i^q$. Thus they act on each $\overline R_d$, and on its dual by $(g\lambda)(f)=\lambda(g^{-1}f)$. The contraction maps satisfy
\[
 \overline C_{gL}(g\lambda)=g\,\overline C_L(\lambda).
\]
For each reduced monomial $x^\alpha$, let
$e_\alpha\in\overline R_{|\alpha|}^*$ be the functional that extracts the coefficient
of $x^\alpha$. These are the monomial coefficient functionals.

\begin{lemma}\label{lem:truncated-monomial-replacement}
For every nonzero subspace $W\subseteq\overline R_D^*$, there is a family
$\calA$ of reduced degree-$D$ monomials such that
$W_\calA=\left\langle e_\alpha:x^\alpha\in\calA\right\rangle$ is
$\calB$-invariant and
$$
 \dim(W_\calA)=\dim(W),\qquad
 \rank(\overline C_{x_1}|_{W_\calA})
 =\bar\rho(W_\calA)\le\bar\rho(W).
$$
Here we use $\overline C_L$ for the contraction map obtained by substituting coefficients of $L$ into $\overline C_{L_z}$.
\end{lemma}
\begin{proof}
The proof is essentially a shifting argument using operations from $\calB$.
Put $h=\dim(W)$, $r=\bar\rho(W)$, and consider
\[
 \calX=\{U\in\operatorname{Gr}(h,\overline R_D^*):\bar\rho(U)\le r\}.
\]
Here the Grassmannian $\operatorname{Gr}(h,\overline R_D^*)$ consists of $h$-dimensional subspaces of $\overline R_D^*$ and is a nonempty projective
variety:\footnote{A projective variety is a nonempty geometric shape defined by the common zeros of a finite set of homogeneous polynomials inside a projective space.}
the set $\calX$ is nonempty as it contains $W$.
The condition $\bar\rho(U)\le r$ is equivalent to the condition that the determinant of every $(r+1)\times(r+1)$ minor of
$\overline C_{L_z}|_U$ vanishes in $\overline{\F_q}(z)$. Their vanishing defines a closed subset, so $\calX$ is closed in the projective Grassmannian and is itself projective.

Since every $g\in\calB$ is a linear invertible change of coordinates, we always have $\bar\rho(gU)=\bar\rho(U)$ and thus $\calX$ is $\calB$-invariant.
The triangular group $\calB$ is connected and solvable, and its induced action on $\calX$ is algebraic. By Borel's fixed-point theorem~\cite[Theorem 10.4]{Borel91}, which states that a connected solvable algebraic group acting algebraically on a nonempty projective variety has a fixed point, there is $U\in\calX$ with $gU=U$ for every $g\in\calB$.

Next, we show that $U$ must be spanned by monomial coefficient functionals.
Assume this is false. Take the smallest set $\calK$ of reduced monomials such that there is a linear functional $\ell=\sum_{\alpha\in\calK}c_\alpha\cdot e_\alpha\in U$ where $c_\alpha\ne0$ for every $\alpha\in\calK$, but for some $\alpha\in\calK$ the monomial coefficient functional $e_\alpha$ is not in $U$. By assumption, we have $|\calK|\ge2$.
Now we consider a diagonal change $g\in\calB$ where $g(x_j)=a_jx_j$ for each $j\in[n]$.
The dual action of such a $g$ on $e_\alpha$ is a scaling by $a^{-\alpha}:=\prod_ja_j^{-\alpha_j}$. Since $\overline{\F_q}$ is infinite, one can pick $g$ such that $a^{-\alpha}$'s are all distinct.
Since $gU=U$ and $\ell\in U$, we know that $g\ell\in U$, which means $\ell'=c\cdot g\ell-\ell\in U$ for every $c\in\overline{\F_q}$. Choose $\alpha_0\in\calK$ with $e_{\alpha_0}\notin U$, choose $\beta\in\calK\setminus\{\alpha_0\}$, and set $c=a^\beta$. Since the weights are distinct, $\ell'$ has support $\calK\setminus\{\beta\}$ and retains a nonzero coefficient at $e_{\alpha_0}$, contradicting the minimality of $\calK$.
In conclusion, $U=W_\calA$ for a family $\calA$ of reduced monomials, where $\dim(W_\calA)=\dim(U)=h=\dim(W)$ and $\bar\rho(W_\calA)=\bar\rho(U)\le r=\bar\rho(W)$. 

Finally, we show that contraction by $L=x_1$ attains the symbolic rank $\bar\rho(U)$. The $\calB$-orbit of $x_1$ consists of all linear forms $L=g(x_1)=\sum_j a_jx_j$ with $a_1\ne0$. Since $U$ is $\calB$-invariant, we have $\rank(\overline C_L|_U)=\rank(\overline C_L|_{gU})=\rank(\overline C_{g^{-1}L}|_U)=\rank(\overline C_{x_1}|_U)$ for every such $L$. If $\bar\rho(U)=0$, the claim is immediate. Otherwise, a minor witnessing $\bar\rho(U)$ is a nonzero polynomial $p(a_1,\ldots,a_n)$. So over the infinite field $\overline{\F_q}$ we can choose $L$ with $a_1\ne0$ and $p(a_1,\ldots,a_n)\ne0$. Consequently, $\rank(\overline C_{x_1}|_{W_\calA})=\bar\rho(U)=\bar\rho(W_\calA)$, which completes the proof.
\end{proof}

The replacement in \Cref{lem:truncated-monomial-replacement} reduces our task to counting monomials: we need to show that a large family $\calA$ with $\calB$-invariant span $W_\calA$ has many members divisible by $x_1^2$, which turns out to be purely combinatorial.

For a family $\calA$ of monomials, its \emph{shadow} is defined as
\[
 \partial\calA=\{x^\alpha/x_i\colon x^\alpha\in\calA,\alpha_i\ge1\}.
\]
In other words, each monomial in $\calA$ is identified by a multiset $\alpha$ and $\partial\calA$ contains all $\alpha'\subset\alpha$ with $|\alpha'|=|\alpha|-1$.

\begin{corollary}\label{cor:to_shadow}
Define $c=\min\{q-1,D\}$.
For every nonzero subspace $W\subseteq\overline R_D^*$, there are families $\calA_d,0\le d\le c$ of reduced degree-$(D-d)$ monomials such that 
\begin{enumerate} 
\item\label{itm:cor:to_shadow_1} $\dim(W)=\sum_{d=0}^c|\calA_d|$,
\item\label{itm:cor:to_shadow_2} $\bar\rho(W)\ge\sum_{d=2}^c|\calA_d|$,
\item\label{itm:cor:to_shadow_3} $\partial\calA_d\subseteq\calA_{d+1}$ holds for every $0\le d\le c-1$.
\end{enumerate}
\end{corollary}
\begin{proof}
By \Cref{lem:truncated-monomial-replacement}, we obtain $\calA$ such that $W_\calA$ is $\calB$-invariant, $\dim(W)=\dim(W_\calA)=|\calA|$, and $\bar\rho(W)\ge\rank(\overline C_{x_1}|_{W_\calA})$.
Partition $\calA$ as $\calA_0,\calA_1,\ldots,\calA_c$ based on the exponent of $x_1$. That is, for $d=0,1,\ldots,c$, define
$$
\calA_d=\left\{x^\alpha/x_1^d\colon x^\alpha\in\calA,\alpha_1=d\right\}.
$$
Note that each $\calA_d$ is a family of reduced degree-$(D-d)$ monomials and $|\calA|=\sum_{d=0}^c|\calA_d|$. This proves \Cref{itm:cor:to_shadow_1}.

Recall the definition of $\overline C_{x_1}$. Its action on $e_{\alpha}\in\overline R_D^*$ produces one of the following:
\begin{itemize}
    \item if $\alpha_1\ge2$, $\overline C_{x_1}(e_\alpha)=e_{\alpha'}$ where $\alpha'_1=\alpha_1-2$ and $\alpha'_j=\alpha_j$ for $2\le j\le n$;
    \item if $\alpha_1\le1$, $\overline C_{x_1}(e_\alpha)=0$.
\end{itemize}
Since each $\alpha'$ in the first case is distinct, $\rank(\overline C_{x_1}|_{W_\calA})=\sum_{d=2}^c|\calA_d|$. This proves \Cref{itm:cor:to_shadow_2}.

Finally we prove \Cref{itm:cor:to_shadow_3} using $\calB$-invariance.
Let $d\le c-1$ and $x^\alpha\in\calA_d$. Define $x^\beta=x^\alpha\cdot x_1^d$; then $x^\beta\in\calA$ by the definition of $\calA_d$.
It remains to show that $x^{\alpha'}\in\calA_{d+1}$ for every $\alpha'\subset\alpha$ of size $|\alpha'|=|\alpha|-1$. Fix an arbitrary $2\le j\le n$ such that $\alpha'$ comes from $\alpha$ by removing one $j$, i.e., $x^{\alpha'}=x^\alpha/x_j$.
Since $d+1\le q-1$, we know that $x^\beta\cdot x_1/x_j=x^\alpha\cdot x_1^{d+1}/x_j$ is a reduced monomial.
Then $x^{\alpha'}\in\calA_{d+1}$ is equivalent to $x^\beta\cdot x_1/x_j\in\calA$. 
To prove this, consider $g\in\calB$ where $g(x_1)=x_1+x_j$ and $g(x_i)=x_i$ for $i\ge2$. 
Then
\begin{align*}
(ge_\beta)\left(x^\beta\cdot x_1/x_j\right)
&=(ge_\beta)\left(x^\alpha\cdot x_1^{d+1}/x_j\right)
\tag{since $x^\beta=x^\alpha\cdot x_1^d$}\\
&=e_\beta\left((x_1-x_j)^{d+1}x^\alpha/x_j\right)
\tag{since $g$ only changes $x_1$ to $x_1+x_j$}\\
&=-(d+1)\ne0
\tag{since $e_\beta$ only isolates the coefficient of $x^\beta=x^\alpha\cdot x_1^d$}
\end{align*}
Since $e_\beta\in W_\calA$ and $W_\calA$ is $\calB$-invariant, $ge_\beta\in W_\calA$.
Since $W_\calA$ is spanned by monomial coefficient functionals and $ge_\beta$ has a nonzero coefficient at the monomial $x^\beta\cdot x_1/x_j$, the monomial $x^\beta\cdot x_1/x_j$ must also lie in $\calA$, which completes the proof.
\end{proof}

\subsubsection{Shadow bounds and putting things together}\label{sec:lwe-progress-shadows}

\Cref{cor:to_shadow} reduces the symbolic rank to a shadow estimation
problem. We use the following consequence of the multiset analogue of
Lov\'asz's shadow bound \cite{BBCH03}.

\begin{theorem}\label{thm:multiset-shadow}
Let $\calA$ be a family of degree-$r$ monomials where $r\ge2$.
Then $|\partial\calA|\ge|\calA|^{1-1/r}$.
\end{theorem}
\begin{proof}
The statement is trivial if $\calA=\emptyset$.
Now assume $N=|\calA|\ge1$. Choose the unique real $t\ge1$ such that $N=\binom{t+r-1}{r}$. Since $\binom{t+r-1}r\ge\left(\frac{t+r-1}r\right)^r$, we have $t\le r\cdot N^{1/r}-r+1$.
By \cite[Theorem~1.7]{BBCH03},
\[
 |\partial\calA|\ge\binom{t+r-2}{r-1}
 =\frac{r\cdot N}{t+r-1}\ge N^{1-1/r}
\]
as claimed.
\end{proof}

The following \Cref{thm:truncated-profile} bounds the symbolic rank in terms of the dimension alone.

\begin{theorem}\label{thm:truncated-profile}
For $D\ge3$, every nonzero subspace $W\subseteq\overline R_D^*$ satisfies
\[
 \bar\rho(W)\ge\left(\frac{\dim(W)}3\right)^{1-2/D}.
\]
\end{theorem}
\begin{proof}
Take $c=\min\{q-1,D\}$ and the families $\calA_0,\ldots,\calA_c$
from \Cref{cor:to_shadow}. Write
\[
 h=\dim(W)
 \quad\text{and}\quad
 a=\sum_{d=2}^c|\calA_d|.
\]
Then \Cref{cor:to_shadow} shows $h=|\calA_0|+|\calA_1|+a$ and $\bar\rho(W)\ge a$.
Since $D\ge3$, \Cref{thm:multiset-shadow} gives
$|\calA_1|\ge|\calA_0|^{1-\frac1D}$ and $|\calA_2|\ge|\calA_1|^{1-\frac1{D-1}}$.
Since $|\calA_2|\le a$, we have
\[
 |\calA_1|\le a^{(D-1)/(D-2)}
 \quad\text{and}\quad 
 |\calA_0|\le a^{D/(D-2)},
\]
which implies
\[
 h=|\calA_0|+|\calA_1|+a
 \le a^{D/(D-2)}+a^{(D-1)/(D-2)}+a
 \le3\cdot a^{D/(D-2)}.
\]
Thus $\bar\rho(W)\ge a\ge(h/3)^{1-2/D}$, which completes the proof.
\end{proof}

Now we are ready to finish the proof of \Cref{prop:conditional-image-growth}.

\begin{proof}[Proof of \Cref{prop:conditional-image-growth}]
Let $H=\dim(R_D)$ and define for $h\in[0,H]$
$$
 \gamma(h)=\left(\frac h3\right)^{1-2/D}.
$$
We first prove the conditional expectation bound.
Fix a history $\calF_{t-1}$ with $h_{t-1}=\dim(W_{t-1})>0$.
Applying \Cref{thm:truncated-profile} to $\overline W_{t-1}$ and using \Cref{fct:rhorho+}, we have
$\rho(W_{t-1})\ge\gamma(h_{t-1})$.
By \Cref{clm:contraction_CL} and \Cref{lem:finite-rank-progress}, we have
\[
 0\le X_t\le\rho(W_{t-1})
 \quad\text{and}\quad
 \E[X_t\mid\calF_{t-1}]\ge\rho(W_{t-1})/3.
\]
Since $\gamma(h_{t-1})\le\rho(W_{t-1})$ and $X_t\le\rho(W_{t-1})$, we also have
\begin{align*}
\E\!\left[\min\!\left\{\frac{X_t}{\gamma(h_{t-1})},1\right\}
       \,\middle|\,\calF_{t-1}\right]
       \ge\E\!\left[\frac{X_t}{\rho(W_{t-1})}
       \,\middle|\,\calF_{t-1}\right]
\ge\frac{\E[X_t\mid\calF_{t-1}]}{\rho(W_{t-1})}\ge\frac13.
\end{align*}
Finally for the integral, observe that
\begin{align*}
\int_0^H\frac{dh}{\gamma(h)}
&\le3\int_0^H h^{-1+2/D}\,dh
 =\frac{3D}{2}\cdot H^{2/D}
 \le\frac{100n^2}{D},
\end{align*}
where for the last inequality we use \Cref{fct:dim-V} to obtain $H\le\binom{n+D}{D}\le\left(\frac{6n}{D}\right)^D$.
\end{proof}

\section{Classical average-case \texorpdfstring{\FqSubsetSum}{Fq-Subset-Sum}}\label{sec:general-p-average}

In this section, we prove \Cref{thm:general-p-average}, which gives a classical algorithm for solving average-case \Cref{prob:fp-subset-sum} and quadratically improves the results in \cite{KOW26}.

\generalpaverage*

The proof of \Cref{thm:general-p-average} relies on three ingredients: (1) an \emph{affine} reducibility notion that upgrades the reducibility notion from \cite{KOW26}, (2) the algorithm from \cite{KOW26} that solves worst-case signed-\FqSubsetSum, and (3) a martingale analysis to control the failure probability.

\paragraph{Section organization.}
\Cref{sec:affine-Fp} defines an affine reducible vector and gives its algorithmic construction.
\Cref{sec:classical-ag-Fq-algorithm} explains the recursive algorithm and proves \Cref{thm:general-p-average}.

\subsection{Affine reducible vector}\label{sec:affine-Fp}

We start with the notion of an affine reducible vector, which generalizes the notion of a reducible vector in \cite{KOW26}. We highlight its mathematical definition and omit the algorithmic feature, which will be apparent during the analysis.

\begin{definition}[Affine reducible vector]\label{def:subset-reducible-line}
Let $v_1,\ldots,v_\ell\in\F_q^d$ be a set of vectors.
For $c\in\F_q^d$ and $0\ne u\in\F_q^d$, we say $(c,u)$ is affine reducible on $v_1,\ldots,v_\ell$ if, for every $a\in\F_q$, we can express $c+a\cdot u$ as a subset-sum of $v_1,\ldots,v_\ell$.
\end{definition}

We remark that if $c=0$, then affine reducibility falls back to the reducibility notion in \cite{KOW26}.
To construct an affine reducible vector, we need the following theorem from \cite{KOW26}.

\begin{theorem}[{\cite{KOW26}}]\label{thm:kow-signed-relation}
Assume $q\ge5$ and let $k=\lfloor(q+3)/4\rfloor$.
Given
\[
 m\ge m_d=(k-1)^{(k-1)(k-2)/2}
 (d+k)^k
\]
vectors $v_1,\ldots,v_m\in\F_q^d$, one can in deterministic $\poly(m,\log q)$ time find $\alpha_1,\ldots,\alpha_m\in\{0,\pm1\}$ such that $\alpha_i$'s are not all zero and $\sum_{i\in[m]}\alpha_i\cdot v_i=0$.
\end{theorem}

Using \Cref{thm:kow-signed-relation}, one can obtain an affine reducible vector $(c,u)$ from random vectors by quotienting by $u$ first, then combining disjoint signed-subset-zero-sums in the smaller space, and finally piecing them together. This is achieved in the following \Cref{prop:general-p-affine-line}.

\begin{proposition}\label{prop:general-p-affine-line}
Let $V=\langle u\rangle\oplus W$ be a vector space over $\F_q$ of dimension $d+1$, where $u\ne0$. Recall $m_d$ from \Cref{thm:kow-signed-relation}.
For every integer $B\ge1$ and given $B\cdot m_d$ independent uniform vectors in $V$, one can in deterministic $\poly(B,m_d,q)$ time construct $c\in V$ such that $(c,u)$ is affine reducible on the input vectors with failure probability at most $q^2/2^B$.
\end{proposition}
\begin{proof}
Let $J_1,\ldots,J_B$ be the disjoint blocks of input vectors, each of size $m_d$.
Write each input vector $v_i$ as
\[
 v_i=\lambda_i u+w_i\quad\text{where}\quad \lambda_i\in\F_q,\quad w_i\in W.
\]
For each $j\in[B]$, apply \Cref{thm:kow-signed-relation} to $(w_i)_{i\in J_j}$ and obtain coefficients $(\alpha_i)_{i\in J_j}\in\{0,\pm1\}^{J_j}$ that are not all zero and satisfy $\sum_{i\in J_j}\alpha_iw_i=0$.  Define
\[
 P_j=\{i\in J_j:\alpha_i=1\},\qquad N_j=\{i\in J_j:\alpha_i=-1\},\qquad
 \beta_j=\sum_{i\in J_j}\alpha_i\lambda_i.
\]
Then
\begin{equation}
 \sum_{i\in P_j}v_i=\sum_{i\in N_j}v_i+\beta_j u,
 \label{eq:general-p-two-point-segment}
\end{equation}
where $\beta_1,\ldots,\beta_B$ are independent and uniform since $\alpha_i$'s depend only on $w_i$'s.

We next show that, with high probability, the subset-sums of the $\beta_j$'s cover all $\F_q$.  For each $0\ne a\in\F_q$, let $X_a$ count $\emptyset\ne T\subseteq[B]$ such that $\sum_{j\in T}\beta_j=a$.
Then a direct calculation shows
$$
\E[X_a]=\frac{2^B-1}{q}
\quad\text{and}\quad
\Var[X_a]=\left(2^B-1\right)\cdot\left(\frac1q-\frac1{q^2}\right)\le\frac{2^B}q.
$$
Therefore by Chebyshev's inequality and a union bound, we have
$$
\Pr\left[\text{subset-sums of $\beta_j$'s cover all }\F_q\right]\ge1-\frac{q^2}{2^B}.
$$

Now it suffices to construct $c\in V$ for affine reducibility, assuming the event above.
To this end, define for each $a\in\F_q$ a set $T_a\subseteq[B]$ such that $\sum_{j\in T_a}\beta_j=a$, where $T_0=\emptyset$.
Let
\[
 c=\sum_{j=1}^B\sum_{i\in N_j}v_i
 \quad\text{and}\quad
 R_a=\bigcup_{j\in T_a}P_j\;\cup\!\bigcup_{j\notin T_a}N_j.
\]
Note that these $P_j,N_j$'s are disjoint.
Observe that 
\begin{align*}
 \sum_{i\in R_a}v_i
 &=\sum_{j\in T_a}\sum_{i\in P_j}v_i+\sum_{j\notin T_a}\sum_{i\in N_j}v_i
 =\sum_{j\in T_a}\left(\beta_ju+\sum_{i\in N_j}v_i\right)+\sum_{j\notin T_a}\sum_{i\in N_j}v_i
 \tag{by \eqref{eq:general-p-two-point-segment}}\\
 &=\left(\sum_{j\in T_a}\beta_j\right)u+\sum_{j=1}^B\sum_{i\in N_j}v_i
 =a\cdot u+c.
 \tag{by the definition of $T_a$ and $c$}
\end{align*}
Thus $(c,u)$ is affine reducible with witnesses $R_a$'s and the runtime is obvious.
\end{proof}

\subsection{The recursive algorithm}\label{sec:classical-ag-Fq-algorithm}

Just like \cite{KOW26}, our algorithm is recursive and sequentially reduces the dimension using (affine) reducible vectors.
Recall $m_d$ from \Cref{thm:kow-signed-relation} and \Cref{prop:general-p-affine-line}.
Set
\begin{equation}
 B=\left\lceil\log_2(4q^2)\right\rceil,
 \qquad T=4n+2.
 \label{eq:general-p-stopping-parameters}
\end{equation}
We allow $T$ construction attempts in total across all recursive levels and assume
\begin{equation}
 m\ge 2n+1+T\cdot Bm_{n-1},
 \label{eq:general-p-adaptive-budget}
\end{equation}
which is dominated by the bound on $m$ in \Cref{thm:general-p-average}.

Now we describe the algorithm.
\begin{definition}[Recursive \FqSubsetSum algorithm]\label{def:general-p-recursive-algorithm}
Use the first $2n$ input vectors to obtain a basis $b_1,\ldots,b_n$ of $\F_q^n$, and report failure if they do not span.
For each $d=0,1,\ldots,n$, define $V_d=\langle b_1,\ldots,b_d\rangle$.
The algorithm defines a subroutine $\Solve(d,t,\{u_i\}_i,K)$, where it looks for a subset-sum of $\{u_i\}_i\subseteq V_d$ that equals $t\in V_d$, where $K$ is a budget allowing repeated calls to \Cref{prop:general-p-affine-line}.

The base case is $d=0$, in which case we simply return the subset $\{1\}$ indicating $u_1=t=0$.
The answer to \FqSubsetSum is precisely $\Solve(n,0,\{v_{2n+1},\ldots,v_m\},T)$, where $v_{2n+1},\ldots,v_m$ are the $m-2n$ random input vectors and $T$ from \Cref{eq:general-p-stopping-parameters} sets the total budget.

For $d\ge1$, $\Solve(d,t,\{u_i\}_i,K)$ is executed as follows.
\begin{enumerate}[label=(\alph*)]
\item\label{itm:def:general-p-recursive-algorithm_1}
Set $u=b_d$ and $W=V_{d-1}$. Then apply \Cref{prop:general-p-affine-line} using the next $Bm_{d-1}$ fresh vectors from $\{u_i\}_i$.
Repeat this up to $K$ times to obtain an affine reducible vector $(c,u)\in V_d\times V_d$.
Report failure and halt if no affine reducible vector is produced even after $K$ attempts.
\item\label{itm:def:general-p-recursive-algorithm_2}
Let $\pi_d:V_d\to V_{d-1}$ be the quotient map induced by $u$.
If the successful attempt is the $r$th, set $K'=K-r$ and execute
$$
\Solve\left(d-1,\pi_d(t-c),\{\pi_d(u_i)\}_{i>r\cdot Bm_{d-1}},K'\right),
$$
which is supposed to return a subset $S'$, supported on indices beyond $r\cdot Bm_{d-1}$, such that $\sum_{i\in S'}\pi_d(u_i)=\pi_d(t-c)$.
Report failure and halt if the recursive call fails.
We find the unique $a\in\F_q$ such that 
\[
t-\sum_{i\in S'}u_i=c+a\cdot u
\]
and a subset $R_a$, supported on indices at most $r\cdot Bm_{d-1}$ and witnessing the affine reducibility of $c+a\cdot u$, such that
$$
\sum_{i\in R_a}u_i=c+a\cdot u.
$$
Then return $S'\cup R_a$.
\end{enumerate}
\end{definition}

Now we prove \Cref{thm:general-p-average}.

\begin{proof}[Proof of \Cref{thm:general-p-average}]
If $q=3$, then the theorem follows from \cite[Theorem 1.2]{KOW26}.
From now on, we assume $q\ge5$.
We first verify that the bound on $m$ in \eqref{eq:general-p-adaptive-budget} ensures that the recursive solver in \Cref{def:general-p-recursive-algorithm} does not run out of vectors.
To see this, the initialization uses $2n$ input vectors, the $d=0$ base case uses $1$ vector, and each of the at most $T$ attempts uses at most $Bm_{n-1}$ vectors. Hence the total amount is at most $2n+1+T\cdot Bm_{n-1}$, upper bounded by the RHS of \eqref{eq:general-p-adaptive-budget}.

Next, we need to show that the RHS of \eqref{eq:general-p-adaptive-budget} is a lower bound for the one in \Cref{thm:general-p-average}. This is a direct calculation using \eqref{eq:general-p-stopping-parameters} and \Cref{thm:kow-signed-relation}, which we omit here.

The runtime is straightforward from the descriptions above and it remains to show that \Cref{def:general-p-recursive-algorithm} fails with probability at most $2^{-n}$.
To this end, it is easy to see that every successful call returns a subset summing to its target. Hence we only need to upper bound (1) the failure of finding a basis in the first $2n$ input vectors and (2) the failure of obtaining all $n$ affine reducible vectors within the total budget of $T$ in \Cref{itm:def:general-p-recursive-algorithm_1}.

For (1), a direct union bound over the $(q^n-1)/(q-1)$ nonzero linear functionals up to scaling gives
\begin{equation}\label{eq:thm:general-p-average_1}
 \Pr[\text{the first $2n$ vectors do not span}]
 \le\frac{q^n-1}{q-1}\cdot q^{-2n}<\frac{q^{-n}}{q-1}\le2^{-n-1}.
\end{equation}
For (2), observe that, conditional on the history and by \Cref{prop:general-p-affine-line}, each repetition in \Cref{itm:def:general-p-recursive-algorithm_1} receives independent uniform vectors and fails with probability at most
\begin{equation}\label{eq:affine-recursive-single-step}
q^2/2^B\le1/4.
\end{equation}
Condition on successful initialization. 
Let $Y_j$ indicate success on the $j$th attempt in \Cref{itm:def:general-p-recursive-algorithm_1}, and set $Y_j=1$ if the algorithm never reaches the $j$th attempt.  By
\eqref{eq:affine-recursive-single-step}, we have
\[
 \E\left[2^{-Y_j}\mid\calF_{j-1}\right]\le\frac34\cdot\frac12+\frac14\cdot1=\frac58
 \quad\text{and}\quad
 \E\!\left[2^{-\sum_{j=1}^T Y_j}\right]\le(5/8)^T\le2^{-T/2},
\]
where $\calF_{j-1}$ records the history before the $j$th attempt.
Markov's inequality then gives
\[
 \Pr\!\left[\sum_{j=1}^T Y_j<n\right]
 \le2^{n-T/2}=2^{-n-1}.
\]
Combining this with \Cref{eq:thm:general-p-average_1} completes the whole proof.
\end{proof}

\section{Classical worst-case \texorpdfstring{\FqSubsetSum}{Fq-Subset-Sum}}
\label{sec:worst}

In this section, we prove \Cref{thm:worst-main}, which is a classical algorithm for solving worst-case \Cref{prob:fp-subset-sum} and improves the $n^{q\log q}$ bound in \cite{IS17,II24,KOW26} to $n^q$.

\worstmain*

Let $\G$ be the set of nonzero elements in $\F_q$. 
If $a=(a_1,\ldots,a_\ell)\in\G^\ell$ and $s\in\G$, let $a(j\gets s)$ denote the tuple obtained by replacing coordinate $j$ with $s$: 
$$
a(j\gets s)=(a_1,\ldots,a_{j-1},s,a_{j+1},\ldots,a_\ell).
$$
The algorithm in \Cref{thm:worst-main} has two steps which we now describe.

First, we arrange the input vectors into pairwise disjoint subset-sums $w_a$, indexed by $a\in\G^\ell$, such that their \emph{weighted} sum equals zero along every line. This is a specialization of the equation array in \cite{IS17} and a simplification of the equation tree in \cite{KOW26}. 

\begin{lemma}\label{lem:axis-array}
For each $\ell\ge1$, $b\in\G^\ell$, and
\[
 m_\ell=(q-1)^{\ell(\ell-1)/2}(n+1)^\ell
\]
input vectors in $\F_q^n$, one can in deterministic $\poly(m_\ell,q)$ time construct pairwise disjoint sets $J_a$ for all $a\in\G^\ell$ such that $J_b\ne\emptyset$ and subset-sums $w_a=\sum_{i\in J_a}v_i$ satisfy 
$$
\sum_{s\in\G}s\cdot w_{a(j\leftarrow s)}=0
$$
for every $j\in[\ell]$ and every fixing of the other coordinates $a_1,\ldots,a_{j-1},a_{j+1},\ldots,a_\ell$.
\end{lemma}
\begin{proof}
We prove by induction on $\ell$.  
For $\ell=1$, take $m_1=n+1$ input vectors and find a nontrivial linear dependence $\sum_i\alpha_iv_i=0$. Choose any $i_0$ with $\alpha_{i_0}\ne0$ and rescale so that $\alpha_{i_0}=b$. Then set
$J_s=\{i:\alpha_i=s\}$ for $s\in\G$ and the statement follows.

Suppose the lemma holds for $\ell$ and now we prove it for $\ell+1$.
Partition the input vectors into
\[
 E:=(q-1)^\ell(n+1)\quad\text{blocks of size}\quad m_\ell.
\]
Define $b'=(b_1,\ldots,b_\ell)$.
On block $k$, construct its corresponding $w_{a'}^{(k)},J_{a'}^{(k)}$ for $a'\in\G^\ell$, where $J_{b'}^{(k)}\ne\emptyset$, and concatenate all $(q-1)^\ell$ many $w_{a'}^{(k)}$ into a long vector $W^{(k)}\in\F_q^{(q-1)^\ell n}$. 
Since $E>(q-1)^\ell n$, we can find a nontrivial linear dependence
\begin{equation}\label{eq:lem:axis-array_1}
\sum_k\alpha_kW^{(k)}=0,
\end{equation}
which, after rescaling, has $\alpha_{k_0}=b_{\ell+1}$ for some $k_0$.
Now for each $a=(a',s)\in\G^{\ell+1}$ where $a'\in\G^\ell$ and $s\in\G$, define
\[
 w_a=\sum_{k:\alpha_k=s}w_{a'}^{(k)},
\]
whose support $J_a$ is the union of the corresponding block supports.
Note that disjointness of $J_a$'s is preserved by the construction and $J_b=J_{(b',b_{\ell+1})}\ne\emptyset$ is due to $\alpha_{k_0}=b_{\ell+1}$ and $J_{b'}^{(k)}\ne\emptyset$ for all $k$.
To check the vanishing condition, we remark that old ones (i.e., the axes selected by $j\le\ell$) are guaranteed by the induction hypothesis and the concatenation operation; the new ones (i.e., the axis selected by $j=\ell+1$) are guaranteed by \eqref{eq:lem:axis-array_1} coordinate-wise.
Finally, the runtime is straightforward from the analysis above.
\end{proof}

The $J_b\ne\emptyset$ condition in \Cref{lem:axis-array} is to ensure the subset-sums are not all empty.

Given \Cref{lem:axis-array}, the second step is to carefully select cells $a\in\G^{q-1}\cong\G^{\G}$ whose sum vanishes.
The signed version of this was done in \cite{KOW26} by enumerating permutations over $\G$ and using their signs, whose vanishing condition is guaranteed by certain exterior-algebra-type cancellation.
For the unsigned version here, we instead enumerate spanning trees over $\F_q$ and directly add them up, whose vanishing condition is guaranteed by the matrix-tree theorem \cite{Chaiken82}.

\begin{lemma}\label{lem:tree-identity}
Fix a rooted tree $\calT=(V,E)$ with vertices $V=\F_q$, root $0\in V$, and edges $E$. For each $0\ne v\in V$, we use $E(v)\in V$ to denote its parent in $\calT$.
Define $a(\calT)\in\G^{\G}$ by setting $a(\calT)_v=v-E(v)$ for each $0\ne v\in V$.

Let $\Tbb$ be the set of all $\calT$ rooted at $0$. Then
$$
\sum_{\calT\in\Tbb}w_{a(\calT)}=0
$$
holds for any subset-sums $w_a$ from \Cref{lem:axis-array} with $\ell=q-1$.
\end{lemma}

The proof of \Cref{lem:tree-identity} is deferred to \Cref{sec:tree}.
Now we put together \Cref{lem:axis-array} and \Cref{lem:tree-identity} to prove \Cref{thm:worst-main}. 

\begin{proof}[Proof of \Cref{thm:worst-main}]
By \Cref{lem:axis-array} with $\ell=q-1=|\G|$ and $b=\bone$, where $\bone\in\G^\ell$ is the all-one vector, we obtain disjoint subset-sums $w_a$.
Then we apply \Cref{lem:tree-identity} and obtain a zero-sum. Note that this zero-sum is a subset-sum since $a(\calT)\ne a(\calT')$ for all distinct $\calT,\calT'\in\Tbb$.
To see that the subset-sum is nonempty, we observe that $J_{\bone}\ne\emptyset$ from \Cref{lem:axis-array} and $\bone=a(\calT_\star)$ where $\calT_\star$ is the tree of the ordered path: $q-1\to q-2\to\cdots\to2\to1\to0$.
Finally, the runtime is straightforward.
\end{proof}

\subsection{A zero-sum identity over rooted trees}\label{sec:tree}

To prove \Cref{lem:tree-identity}, we need some additional algebraic definitions.

Let $U=\F_q^{\G}$ with standard basis $(e_s)_{s\in[q-1]}$.
Define
\[
 \br=\sum_{s\in\G}s\cdot e_s,
 \qquad \bar U=U/\langle\br\rangle,
\]
and write $\bar e_s$ for the image of $e_s$ in $\bar U$.
The point of quotienting by $\langle\br\rangle$ is that it represents the vanishing conditions in \Cref{lem:axis-array}. This abstraction allows us to omit algorithmic details in \Cref{lem:axis-array} and focus purely on the algebraic conditions. In particular, it suffices to prove the following result.

\begin{proposition}\label{prop:tree}
$\sum_{\calT\in\Tbb}\bigotimes_{v\in\G}\bar e_{a(\calT)_v}=0$ in ${\bar U}^{\otimes\G}$.
\end{proposition}

Indeed, given \Cref{prop:tree}, one can directly obtain \Cref{lem:tree-identity}.

\begin{proof}[Proof of \Cref{lem:tree-identity}]
Define a linear map $\Psi:U^{\otimes\G}\to\F_q^n$ by $\Psi\!\left(\bigotimes_{v\in\G}e_{a_v}\right)=w_a$ on basis tensors and extend it linearly.
\Cref{lem:axis-array} implies that $\Psi$ annihilates every tensor with \(\br\) in any coordinate, i.e., for every $j\in\G$ and $a\in\G^{\G\setminus\{j\}}$, we have
$$
\Psi\left(\bigotimes_{v\in\G}x_v\right)=0,
\qquad\text{where }x_j=\br\text{ and }x_v=e_{a_v}\ \text{for }v\ne j.
$$
Hence $\Psi$ factors through ${\bar U}^{\otimes\G}$: there exists a linear map $\bar\Psi\colon{\bar U}^{\otimes\G}\to\F_q^n$ such that $\Psi=\bar\Psi\circ\Pi$, where $\Pi=\pi^{\otimes\G}$ and $\pi\colon U\to\bar U$ is the quotient map that linearly takes $e_s$ to $\bar e_s$.
Consequently, we have
\begin{align*}
\sum_{\calT\in\Tbb}w_{a(\calT)}
&=\sum_{\calT\in\Tbb}\Psi\left(\bigotimes_{v\in\G}e_{a(\calT)_v}\right)
=\sum_{\calT\in\Tbb}\bar\Psi\circ\Pi\left(\bigotimes_{v\in\G}e_{a(\calT)_v}\right)\\
&=\sum_{\calT\in\Tbb}\bar\Psi\left(\bigotimes_{v\in\G}\bar e_{a(\calT)_v}\right)=\bar\Psi\left(\sum_{\calT\in\Tbb}\bigotimes_{v\in\G}\bar e_{a(\calT)_v}\right)
\tag{by linearity}\\
&=\bar\Psi(0)=0
\tag{by \Cref{prop:tree}}
\end{align*}
as desired.
\end{proof}

It remains to prove the pure algebraic formula \Cref{prop:tree}.

\begin{proof}[Proof of \Cref{prop:tree}]
To prove the identity, we test it against an explicit basis of ${\bar U}^*$, the dual space of $\bar U$, and the directed matrix-tree theorem \cite{Chaiken82} will turn each resulting weighted tree sum into a singular determinant.
We recall the precise weighted form of the theorem used below.

\begin{theorem}[{\cite{Chaiken82}}]\label{thm:matrixtree}
Let $V$ be a finite vertex set with a distinguished root $\rtsf\in V$.
For each $u,v\in V$, assign a weight $c_{v,u}$ in a commutative ring $R$.  Define the out-Laplacian $L$ by
\[
 L_{v,u}=\begin{cases}
 \displaystyle\sum_{z\in V\setminus\{v\}}c_{v,z}&u=v,\\
 -c_{v,u}&u\ne v.
 \end{cases}
\]
Let $L'$ be obtained by deleting row $\rtsf$ and column $\rtsf$ from $L$.
Let\footnote{Each $\calT=(V,E)\in\Tbb$ has the following property. The root $\rtsf$ has no outgoing edge. Each $v\ne\rtsf$ has one outgoing edge $E(v)$, and repeatedly following these edges reaches $\rtsf$.} $\Tbb$ be the set of spanning trees of $V$ rooted at $\rtsf$. Then
\[
 \det(L')
 =\sum_{\calT=(V,E)\in\Tbb}\prod_{v\in V\setminus\{\rtsf\}}c_{v,E(v)}.
\]
\end{theorem}

For $0\le k\le q-3$, define a linear map $\lambda_k\colon U\to\F_q$ by $\lambda_k(e_s)=s^k$ on basis vectors and extend it linearly.
These maps annihilate $\br$:
\[
 \lambda_k(\br)=\sum_{s\in\G}s^{k+1}=0
 \quad\text{holds for every $1\le k+1\le q-2$.}
\]
To see this, observe that for every $t\in\G$, we have $t^{k+1}\lambda_k(\br)=\sum_{s\in\G}(st)^{k+1}=\sum_{s\in\G}s^{k+1}=\lambda_k(\br)$, which means $(t^{k+1}-1)\lambda_k(\br)=0$; when $k+1\le q-2$, one can always pick $t\in\G$ such that $t^{k+1}-1\ne0$ and hence $\lambda_k(\br)=0$.
Consequently, each $\lambda_k$ factors through $\bar U$: there exists a linear map $\bar\lambda_k\colon\bar U\to\F_q$ such that $\lambda_k=\bar\lambda_k\circ\pi$ where $\pi\colon U\to\bar U$ is the quotient map induced by $\br$.

On the other hand, these $\bar\lambda_k$'s are linearly independent.
This is because any linear combination of them lifts to a linear combination of $\lambda_k$, which corresponds to a polynomial of degree at most $q-3$ and cannot vanish on all $s\in\G$ unless it is zero.\footnote{We remark that one can define $\lambda_k$ for $k=q-2$ analogously and show that $\lambda_0,\lambda_1,\ldots,\lambda_{q-2}$ are linearly independent. However $\lambda_{q-2}$ does not factor through $\bar U$: $\sum_{s\in\G}s^{q-1}=q-1\ne0$; hence we do not include it to study $\bar U$. In other words, $\lambda_0,\lambda_1,\ldots,\lambda_{q-2}$ form a basis of $U^*$, but since we care about ${\bar U}^*$ for \Cref{prop:tree}, we should use $\bar\lambda_0,\bar\lambda_1,\ldots,\bar\lambda_{q-3}$.} Since $\dim(\bar U)=\dim(U)-1=|\G|-1=q-2$, these $\bar\lambda_k$ form a basis of
${\bar U}^*$.

Going back to \Cref{prop:tree}, it therefore suffices to prove for every $\ksf=(\ksf_v)_{v\in\G}\in\{0,\ldots,q-3\}^{\G}$
$$
M(\ksf):=\left(\bigotimes_{v\in\G}\lambda_{\ksf_v}\right)\left(\sum_{\calT\in\Tbb}\bigotimes_{v\in\G}e_{a(\calT)_v}\right)=\sum_{\calT=(V,E)\in\Tbb}\prod_{v\in\G}
 (v-E(v))^{\ksf_v}=0.
$$
Apply \Cref{thm:matrixtree} with $V=\F_q$, $\rtsf=0$, $R=\F_q$, and $c_{v,u}=(v-u)^{\ksf_v}$, setting $\ksf_0=0$. Then $M(\ksf)=\det(L')$, whose rows and columns are indexed by $\G$ and
\[
 (L')_{v,u}=\begin{cases}
 \sum_{z\in\F_q\setminus\{v\}}(v-z)^{\ksf_v}&u=v\\
 -(v-u)^{\ksf_v}&u\ne v
 \end{cases}
 =-(v-u)^{\ksf_v},
\]
where we used the fact that 
$$
\sum_{z\in\F_q\setminus\{v\}}(v-z)^{\ksf_v}=\begin{cases}
0&1\le\ksf_v\le q-3,\\
-1&\ksf_v=0.
\end{cases}
$$
It remains to show that $L'$ is not full rank. To this end, for each $v\in\G$, define a polynomial $f_v(x)=(v-x)^{\ksf_v}$. Then $(L')_{v,u}=-f_v(u)$.
Since each $f_v$ has degree at most $q-3$ yet $|\G|=q-1$, these $f_v$'s are linearly dependent, which implies linear dependence of rows of $L'$. This completes the whole proof.
\end{proof}

\bibliographystyle{alphaurl}
\bibliography{ref}

\newcommand{\etalchar}[1]{$^{#1}$}
\begin{thebibliography}{BBWZW26}

\bibitem[ACF{\etalchar{+}}15a]{ACFP15}
Martin~R. Albrecht, Carlos Cid, Jean-Charles Faug{\`e}re, Robert Fitzpatrick,
  and Ludovic Perret.
\newblock Algebraic algorithms for {LWE} problems.
\newblock {\em ACM Communications in Computer Algebra}, 49(2):62, 2015.
\newblock \href {https://doi.org/10.1145/2815111.2815158}
  {\path{doi:10.1145/2815111.2815158}}.

\bibitem[ACF{\etalchar{+}}15b]{ACFFP15BKW}
Martin~R. Albrecht, Carlos Cid, Jean-Charles Faug{\`e}re, Robert Fitzpatrick,
  and Ludovic Perret.
\newblock On the complexity of the {BKW} algorithm on {LWE}.
\newblock {\em Designs, Codes and Cryptography}, 74(2):325--354, 2015.
\newblock URL: \url{https://eprint.iacr.org/2012/636}, \href
  {https://doi.org/10.1007/s10623-013-9864-x}
  {\path{doi:10.1007/s10623-013-9864-x}}.

\bibitem[ACFP14]{ACFP14}
Martin~R. Albrecht, Carlos Cid, Jean-Charles Faug{\`e}re, and Ludovic Perret.
\newblock Algebraic algorithms for {LWE}.
\newblock Cryptology ePrint Archive, Report 2014/1018, 2014.
\newblock URL: \url{https://eprint.iacr.org/2014/1018}.

\bibitem[AG11]{AG11}
Sanjeev Arora and Rong Ge.
\newblock New algorithms for learning in presence of errors.
\newblock In {\em Automata, Languages and Programming -- 38th International
  Colloquium, ICALP 2011}, volume 6755 of {\em Lecture Notes in Computer
  Science}, pages 403--415. Springer, 2011.
\newblock \href {https://doi.org/10.1007/978-3-642-22006-7_34}
  {\path{doi:10.1007/978-3-642-22006-7_34}}.

\bibitem[BBCH03]{BBCH03}
Airat Bekmetjev, Graham~R. Brightwell, Andrzej Czygrinow, and Glenn Hurlbert.
\newblock Thresholds for families of multisets, with an application to graph
  pebbling.
\newblock {\em Discrete Mathematics}, 269(1--3):21--34, 2003.
\newblock \href {https://arxiv.org/abs/math/0406068}
  {\path{arXiv:math/0406068}}, \href
  {https://doi.org/10.1016/S0012-365X(02)00745-8}
  {\path{doi:10.1016/S0012-365X(02)00745-8}}.

\bibitem[BB{\O}24]{BBO24}
Ward Beullens, Pierre Briaud, and Morten {\O}ygarden.
\newblock A security analysis of restricted syndrome decoding problems.
\newblock {\em {IACR} Communications in Cryptology}, 1(3), 2024.
\newblock URL: \url{https://eprint.iacr.org/2024/611}, \href
  {https://doi.org/10.62056/a06cy7qiu} {\path{doi:10.62056/a06cy7qiu}}.

\bibitem[BBWZW26]{BBWW26}
Sebastian Bitzer, Michele Battagliola, Antonia Wachter-Zeh, and Violetta Weger.
\newblock {TCitH}- and {VOLEitH}-based signatures from restricted decoding.
\newblock In {\em {IEEE} International Symposium on Information Theory, {ISIT}
  2026}, pages 1--6. {IEEE}, 2026.
\newblock \href {https://arxiv.org/abs/2510.11224} {\path{arXiv:2510.11224}},
  \href {https://doi.org/10.1109/ISIT62367.2026.11654135}
  {\path{doi:10.1109/ISIT62367.2026.11654135}}.

\bibitem[BCC{\etalchar{+}}23]{Wave23}
Gustavo Banegas, K{\'e}vin Carrier, Andr{\'e} Chailloux, Alain Couvreur, Thomas
  Debris-Alazard, Philippe Gaborit, Pierre Karpman, Johanna Loyer, Ruben
  Niederhagen, Nicolas Sendrier, Benjamin Smith, and Jean-Pierre Tillich.
\newblock {WAVE} specification document, June 2023.
\newblock Round 1 submission, Version 1, released June 1, 2023.
\newblock URL:
  \url{https://csrc.nist.gov/csrc/media/Projects/pqc-dig-sig/documents/round-1/spec-files/wave-spec-web.pdf}.

\bibitem[BCDAL20]{BCDL19}
R{\'e}mi Bricout, Andr{\'e} Chailloux, Thomas Debris-Alazard, and Matthieu
  Lequesne.
\newblock Ternary syndrome decoding with large weight.
\newblock In {\em Selected Areas in Cryptography -- SAC 2019}, volume 11959 of
  {\em Lecture Notes in Computer Science}, pages 437--466. Springer, 2020.
\newblock URL: \url{https://eprint.iacr.org/2019/304}, \href
  {https://doi.org/10.1007/978-3-030-38471-5_18}
  {\path{doi:10.1007/978-3-030-38471-5_18}}.

\bibitem[BCT26]{BCT25}
Agathe Blanvillain, Andr{\'e} Chailloux, and Jean-Pierre Tillich.
\newblock The quantum decoding problem: Tight achievability bounds and
  application to {Regev}'s reduction.
\newblock {\em IEEE Transactions on Information Theory}, 72(7):4980--4999,
  2026.
\newblock \href {https://arxiv.org/abs/2509.24796} {\path{arXiv:2509.24796}},
  \href {https://doi.org/10.1109/TIT.2026.3689737}
  {\path{doi:10.1109/TIT.2026.3689737}}.

\bibitem[BGG{\etalchar{+}}16]{BGGOP16}
Johannes Buchmann, Florian G{\"o}pfert, Tim G{\"u}neysu, Tobias Oder, and
  Thomas P{\"o}ppelmann.
\newblock High-performance and lightweight lattice-based public-key encryption.
\newblock In {\em Proceedings of the 2nd ACM International Workshop on IoT
  Privacy, Trust, and Security}, pages 2--9. ACM, 2016.
\newblock \href {https://doi.org/10.1145/2899007.2899011}
  {\path{doi:10.1145/2899007.2899011}}.

\bibitem[BGPW16]{BGPW16}
Johannes Buchmann, Florian G{\"o}pfert, Rachel Player, and Thomas Wunderer.
\newblock On the hardness of {LWE} with binary error: Revisiting the hybrid
  lattice-reduction and meet-in-the-middle attack.
\newblock In {\em Progress in Cryptology -- AFRICACRYPT 2016}, volume 9646 of
  {\em Lecture Notes in Computer Science}, pages 24--43. Springer, 2016.
\newblock URL: \url{https://eprint.iacr.org/2016/089}, \href
  {https://doi.org/10.1007/978-3-319-31517-1_2}
  {\path{doi:10.1007/978-3-319-31517-1_2}}.

\bibitem[BKW03]{BKW03}
Avrim Blum, Adam Kalai, and Hal Wasserman.
\newblock Noise-tolerant learning, the parity problem, and the statistical
  query model.
\newblock {\em Journal of the ACM}, 50(4):506--519, 2003.
\newblock \href {https://doi.org/10.1145/792538.792543}
  {\path{doi:10.1145/792538.792543}}.

\bibitem[BL24]{baily2024strength}
Benjamin Baily and Amichai Lampert.
\newblock Strength is bounded linearly by {Birch} rank.
\newblock {\em arXiv preprint arXiv:2410.00248}, 2024.
\newblock \href {https://arxiv.org/abs/2410.00248} {\path{arXiv:2410.00248}}.

\bibitem[Bor91]{Borel91}
Armand Borel.
\newblock {\em Linear Algebraic Groups}, volume 126 of {\em Graduate Texts in
  Mathematics}.
\newblock Springer, second edition, 1991.
\newblock \href {https://doi.org/10.1007/978-1-4612-0941-6}
  {\path{doi:10.1007/978-1-4612-0941-6}}.

\bibitem[BPW{\etalchar{+}}23]{BBPSWW23}
Sebastian Bitzer, Alessio Pavoni, Violetta Weger, Paolo Santini, Marco Baldi,
  and Antonia Wachter-Zeh.
\newblock Generic decoding of restricted errors.
\newblock In {\em 2023 {IEEE} International Symposium on Information Theory},
  pages 246--251. {IEEE}, 2023.
\newblock \href {https://arxiv.org/abs/2303.08882} {\path{arXiv:2303.08882}},
  \href {https://doi.org/10.1109/ISIT54713.2023.10206983}
  {\path{doi:10.1109/ISIT54713.2023.10206983}}.

\bibitem[CDAE21]{CDE21}
Andr{\'e} Chailloux, Thomas Debris-Alazard, and Simona Etinski.
\newblock Classical and quantum algorithms for generic syndrome decoding
  problems and applications to the {Lee} metric.
\newblock In {\em Post-Quantum Cryptography -- {PQCrypto} 2021}, volume 12841
  of {\em Lecture Notes in Computer Science}, pages 44--62. Springer, 2021.
\newblock \href {https://arxiv.org/abs/2104.12810} {\path{arXiv:2104.12810}},
  \href {https://doi.org/10.1007/978-3-030-81293-5_3}
  {\path{doi:10.1007/978-3-030-81293-5_3}}.

\bibitem[Cha82]{Chaiken82}
Seth Chaiken.
\newblock A combinatorial proof of the all minors matrix tree theorem.
\newblock {\em SIAM Journal on Algebraic and Discrete Methods}, 3(3):319--329,
  1982.
\newblock \href {https://doi.org/10.1137/0603033} {\path{doi:10.1137/0603033}}.

\bibitem[CKPS00]{CKPS00}
Nicolas Courtois, Alexander Klimov, Jacques Patarin, and Adi Shamir.
\newblock Efficient algorithms for solving overdefined systems of multivariate
  polynomial equations.
\newblock In {\em Advances in Cryptology -- EUROCRYPT 2000}, volume 1807 of
  {\em Lecture Notes in Computer Science}, pages 392--407. Springer, 2000.
\newblock \href {https://doi.org/10.1007/3-540-45539-6_27}
  {\path{doi:10.1007/3-540-45539-6_27}}.

\bibitem[CLZ22]{CLZ22}
Yilei Chen, Qipeng Liu, and Mark Zhandry.
\newblock Quantum algorithms for variants of average-case lattice problems via
  filtering.
\newblock In {\em Advances in Cryptology -- EUROCRYPT 2022, Part III}, volume
  13277 of {\em Lecture Notes in Computer Science}, pages 372--401. Springer,
  2022.
\newblock \href {https://arxiv.org/abs/2108.11015} {\path{arXiv:2108.11015}},
  \href {https://doi.org/10.1007/978-3-031-07082-2_14}
  {\path{doi:10.1007/978-3-031-07082-2_14}}.

\bibitem[CT24]{CT23}
Andr{\'e} Chailloux and Jean-Pierre Tillich.
\newblock The quantum decoding problem.
\newblock In {\em 19th Conference on the Theory of Quantum Computation,
  Communication and Cryptography (TQC 2024)}, volume 310 of {\em Leibniz
  International Proceedings in Informatics (LIPIcs)}, pages 6:1--6:14. Schloss
  Dagstuhl -- Leibniz-Zentrum f{\"u}r Informatik, 2024.
\newblock \href {https://arxiv.org/abs/2310.20651} {\path{arXiv:2310.20651}},
  \href {https://doi.org/10.4230/LIPIcs.TQC.2024.6}
  {\path{doi:10.4230/LIPIcs.TQC.2024.6}}.

\bibitem[CT25]{CT24}
Andr{\'e} Chailloux and Jean-Pierre Tillich.
\newblock Quantum advantage from soft decoders.
\newblock In {\em Proceedings of the 57th Annual ACM Symposium on Theory of
  Computing}, pages 738--749. ACM, 2025.
\newblock \href {https://arxiv.org/abs/2411.12553} {\path{arXiv:2411.12553}},
  \href {https://doi.org/10.1145/3717823.3718319}
  {\path{doi:10.1145/3717823.3718319}}.

\bibitem[DART24]{DRT24}
Thomas Debris-Alazard, Maxime Remaud, and Jean-Pierre Tillich.
\newblock Quantum reduction of finding short code vectors to the decoding
  problem.
\newblock {\em IEEE Transactions on Information Theory}, 70(7):5323--5342,
  2024.
\newblock \href {https://arxiv.org/abs/2106.02747} {\path{arXiv:2106.02747}},
  \href {https://doi.org/10.1109/TIT.2023.3327759}
  {\path{doi:10.1109/TIT.2023.3327759}}.

\bibitem[DAST19]{Wave19}
Thomas Debris-Alazard, Nicolas Sendrier, and Jean-Pierre Tillich.
\newblock Wave: A new family of trapdoor one-way preimage sampleable functions
  based on codes.
\newblock In {\em Advances in Cryptology -- ASIACRYPT 2019, Part I}, volume
  11921 of {\em Lecture Notes in Computer Science}, pages 21--51. Springer,
  2019.
\newblock \href {https://doi.org/10.1007/978-3-030-34578-5_2}
  {\path{doi:10.1007/978-3-030-34578-5_2}}.

\bibitem[DEL25]{DEL25}
L{\'e}o Ducas, Lynn Engelberts, and Johanna Loyer.
\newblock Wagner's algorithm provably runs in subexponential time for
  {$\mathrm{SIS}^{\infty}$}.
\newblock In {\em Advances in Cryptology -- CRYPTO 2025, Part I}, volume 16000
  of {\em Lecture Notes in Computer Science}, pages 353--384. Springer, 2025.
\newblock URL: \url{https://eprint.iacr.org/2025/575}, \href
  {https://doi.org/10.1007/978-3-032-01855-7_12}
  {\path{doi:10.1007/978-3-032-01855-7_12}}.

\bibitem[DKSS13]{DKSS13}
Zeev Dvir, Swastik Kopparty, Shubhangi Saraf, and Madhu Sudan.
\newblock Extensions to the method of multiplicities, with applications to
  {Kakeya} sets and mergers.
\newblock {\em SIAM Journal on Computing}, 42(6):2305--2328, 2013.
\newblock \href {https://doi.org/10.1137/100783704}
  {\path{doi:10.1137/100783704}}.

\bibitem[DM73]{diderrich1973combinatorial}
George~T Diderrich and Henry~B Mann.
\newblock Combinatorial problems in finite abelian groups.
\newblock In {\em A survey of combinatorial theory}, pages 95--100.
  North-Holland, 1973.
\newblock \href {https://doi.org/10.1016/B978-0-7204-2262-7.50015-6}
  {\path{doi:10.1016/B978-0-7204-2262-7.50015-6}}.

\bibitem[FIM{\etalchar{+}}14]{FIMS03}
Katalin Friedl, G{\'a}bor Ivanyos, Fr{\'e}d{\'e}ric Magniez, Miklos Santha, and
  Pranab Sen.
\newblock Hidden translation and {Translating Coset} in quantum computing.
\newblock {\em SIAM Journal on Computing}, 43(1):1--24, 2014.
\newblock \href {https://arxiv.org/abs/quant-ph/0211091}
  {\path{arXiv:quant-ph/0211091}}, \href {https://doi.org/10.1137/130907203}
  {\path{doi:10.1137/130907203}}.

\bibitem[Fr{\"o}85]{froberg1985inequality}
Ralf Fr{\"o}berg.
\newblock An inequality for {Hilbert} series of graded algebras.
\newblock {\em Mathematica Scandinavica}, 56(2):117--144, 1985.
\newblock \href {https://doi.org/10.7146/math.scand.a-12092}
  {\path{doi:10.7146/math.scand.a-12092}}.

\bibitem[GJLS21]{GJLS21}
Romain Gay, Aayush Jain, Huijia Lin, and Amit Sahai.
\newblock Indistinguishability obfuscation from simple-to-state hard problems:
  New assumptions, new techniques, and simplification.
\newblock In {\em Advances in Cryptology -- EUROCRYPT 2021, Part III}, volume
  12698 of {\em Lecture Notes in Computer Science}, pages 97--126. Springer,
  2021.
\newblock URL: \url{https://eprint.iacr.org/2020/764}, \href
  {https://doi.org/10.1007/978-3-030-77883-5_4}
  {\path{doi:10.1007/978-3-030-77883-5_4}}.

\bibitem[II24]{II24}
Muhammad Imran and G{\'a}bor Ivanyos.
\newblock Zero sum subsequences and hidden subgroups.
\newblock {\em Quantum Information Processing}, 23(1):14, 2024.
\newblock \href {https://arxiv.org/abs/2304.08376} {\path{arXiv:2304.08376}},
  \href {https://doi.org/10.1007/s11128-023-04228-2}
  {\path{doi:10.1007/s11128-023-04228-2}}.

\bibitem[IN96]{IN96}
Russell Impagliazzo and Moni Naor.
\newblock Efficient cryptographic schemes provably as secure as subset sum.
\newblock {\em Journal of Cryptology}, 9(4):199--216, 1996.
\newblock URL: \url{https://www.wisdom.weizmann.ac.il/~naor/PAPERS/subset.pdf},
  \href {https://doi.org/10.1007/BF00189260} {\path{doi:10.1007/BF00189260}}.

\bibitem[IPS18]{IPS18}
G{\'a}bor Ivanyos, Anupam Prakash, and Miklos Santha.
\newblock On learning linear functions from subset and its applications in
  quantum computing.
\newblock In {\em 26th Annual European Symposium on Algorithms (ESA 2018)},
  volume 112 of {\em Leibniz International Proceedings in Informatics
  (LIPIcs)}, pages 66:1--66:14. Schloss Dagstuhl -- Leibniz-Zentrum f{\"u}r
  Informatik, 2018.
\newblock \href {https://arxiv.org/abs/1806.09660} {\path{arXiv:1806.09660}},
  \href {https://doi.org/10.4230/LIPIcs.ESA.2018.66}
  {\path{doi:10.4230/LIPIcs.ESA.2018.66}}.

\bibitem[IS17]{IS17}
G{\'a}bor Ivanyos and Miklos Santha.
\newblock Solving systems of diagonal polynomial equations over finite fields.
\newblock {\em Theoretical Computer Science}, 657:73--85, 2017.
\newblock \href {https://arxiv.org/abs/1503.09016} {\path{arXiv:1503.09016}},
  \href {https://doi.org/10.1016/j.tcs.2016.04.045}
  {\path{doi:10.1016/j.tcs.2016.04.045}}.

\bibitem[ISS12]{ISS12}
G{\'a}bor Ivanyos, Luc Sanselme, and Miklos Santha.
\newblock An efficient quantum algorithm for the hidden subgroup problem in
  nil-2 groups.
\newblock {\em Algorithmica}, 62(1--2):480--498, 2012.
\newblock \href {https://arxiv.org/abs/0707.1260} {\path{arXiv:0707.1260}},
  \href {https://doi.org/10.1007/s00453-010-9467-0}
  {\path{doi:10.1007/s00453-010-9467-0}}.

\bibitem[Iva08]{Ivanyos07}
G{\'a}bor Ivanyos.
\newblock On solving systems of random linear disequations.
\newblock {\em Quantum Information and Computation}, 8(6--7):579--594, 2008.
\newblock \href {https://arxiv.org/abs/0704.2988} {\path{arXiv:0704.2988}},
  \href {https://doi.org/10.26421/QIC8.6-7-2} {\path{doi:10.26421/QIC8.6-7-2}}.

\bibitem[JSW{\etalchar{+}}25]{DQI}
Stephen~P Jordan, Noah Shutty, Mary Wootters, Adam Zalcman, Alexander
  Schmidhuber, Robbie King, Sergei~V Isakov, Tanuj Khattar, and Ryan Babbush.
\newblock Optimization by decoded quantum interferometry.
\newblock {\em Nature}, 646(8086):831--836, 2025.

\bibitem[KF15]{KF15}
Paul Kirchner and Pierre-Alain Fouque.
\newblock An improved {BKW} algorithm for {LWE} with applications to
  cryptography and lattices.
\newblock In {\em Advances in Cryptology -- CRYPTO 2015, Part I}, volume 9215
  of {\em Lecture Notes in Computer Science}, pages 43--62. Springer, 2015.
\newblock \href {https://arxiv.org/abs/1506.02717} {\path{arXiv:1506.02717}},
  \href {https://doi.org/10.1007/978-3-662-47989-6_3}
  {\path{doi:10.1007/978-3-662-47989-6_3}}.

\bibitem[KL22]{KL22}
Pierre Karpman and Charlotte Lefevre.
\newblock Time-memory tradeoffs for large-weight syndrome decoding in ternary
  codes.
\newblock In {\em Public-Key Cryptography -- {PKC} 2022, Part I}, volume 13177
  of {\em Lecture Notes in Computer Science}, pages 82--111. Springer, 2022.
\newblock URL:
  \url{https://www.iacr.org/archive/pkc2022/131770046/131770046.pdf}, \href
  {https://doi.org/10.1007/978-3-030-97121-2_4}
  {\path{doi:10.1007/978-3-030-97121-2_4}}.

\bibitem[KOW26]{KOW26}
Robin Kothari, Ryan O'Donnell, and Kewen Wu.
\newblock No exponential quantum speedup for {$\mathrm{SIS}^{\infty}$} anymore.
\newblock In {\em Proceedings of the 58th Annual ACM Symposium on Theory of
  Computing}, pages 101--105. ACM, 2026.
\newblock Full version available on arXiv.
\newblock \href {https://arxiv.org/abs/2510.07515} {\path{arXiv:2510.07515}},
  \href {https://doi.org/10.1145/3798129.3800731}
  {\path{doi:10.1145/3798129.3800731}}.

\bibitem[KS99]{KS99}
Aviad Kipnis and Adi Shamir.
\newblock Cryptanalysis of the {HFE} public key cryptosystem by
  relinearization.
\newblock In {\em Advances in Cryptology -- CRYPTO '99}, volume 1666 of {\em
  Lecture Notes in Computer Science}, pages 19--30. Springer, 1999.
\newblock \href {https://doi.org/10.1007/3-540-48405-1_2}
  {\path{doi:10.1007/3-540-48405-1_2}}.

\bibitem[Lov93]{Lovasz93}
L{\'a}szl{\'o} Lov{\'a}sz.
\newblock {\em Combinatorial Problems and Exercises}.
\newblock North-Holland, Amsterdam, 1993.

\bibitem[MP13]{MP13}
Daniele Micciancio and Chris Peikert.
\newblock Hardness of {SIS} and {LWE} with small parameters.
\newblock In {\em Advances in Cryptology -- CRYPTO 2013, Part I}, volume 8042
  of {\em Lecture Notes in Computer Science}, pages 21--39. Springer, 2013.
\newblock \href {https://doi.org/10.1007/978-3-642-40041-4_2}
  {\path{doi:10.1007/978-3-642-40041-4_2}}.

\bibitem[Nen17]{nenashev2017note}
Gleb Nenashev.
\newblock A note on fr{\"o}berg's conjecture for forms of equal degrees.
\newblock {\em Comptes Rendus Mathematique}, 355(3):272--276, 2017.

\bibitem[NMS{\"U}25]{CMSU25}
Miguel~Cueto Noval, Simon-Philipp Merz, Patrick St{\"a}hlin, and Ak{\i}n
  {\"U}nal.
\newblock On the soundness of algebraic attacks against code-based assumptions.
\newblock In {\em Advances in Cryptology -- EUROCRYPT 2025, Part VI}, volume
  15606 of {\em Lecture Notes in Computer Science}, pages 385--415. Springer,
  2025.
\newblock URL: \url{https://eprint.iacr.org/2025/415}, \href
  {https://doi.org/10.1007/978-3-031-91095-1_14}
  {\path{doi:10.1007/978-3-031-91095-1_14}}.

\bibitem[Ols69]{Olson69}
John~E. Olson.
\newblock A combinatorial problem on finite abelian groups, {I}.
\newblock {\em Journal of Number Theory}, 1(1):8--10, 1969.
\newblock \href {https://doi.org/10.1016/0022-314X(69)90021-3}
  {\path{doi:10.1016/0022-314X(69)90021-3}}.

\bibitem[Reg09]{Reg09}
Oded Regev.
\newblock On lattices, learning with errors, random linear codes, and
  cryptography.
\newblock {\em Journal of the ACM}, 56(6):34:1--34:40, 2009.
\newblock \href {https://doi.org/10.1145/1568318.1568324}
  {\path{doi:10.1145/1568318.1568324}}.

\bibitem[Sal23]{salizzoni2023upper}
Flavio Salizzoni.
\newblock An upper bound for the solving degree in terms of the degree of
  regularity.
\newblock {\em arXiv preprint arXiv:2304.13485}, 2023.

\bibitem[SSTX09]{SSTX09}
Damien Stehl{\'e}, Ron Steinfeld, Keisuke Tanaka, and Keita Xagawa.
\newblock Efficient public key encryption based on ideal lattices.
\newblock In {\em Advances in Cryptology -- ASIACRYPT 2009}, volume 5912 of
  {\em Lecture Notes in Computer Science}, pages 617--635. Springer, 2009.
\newblock URL: \url{https://eprint.iacr.org/2009/285}, \href
  {https://doi.org/10.1007/978-3-642-10366-7_36}
  {\path{doi:10.1007/978-3-642-10366-7_36}}.

\bibitem[ST21]{SemaevTenti21}
Igor Semaev and Andrea Tenti.
\newblock Probabilistic analysis on {Macaulay} matrices over finite fields and
  complexity of constructing {Gr{\"o}bner} bases.
\newblock {\em Journal of Algebra}, 565:651--674, 2021.
\newblock \href {https://doi.org/10.1016/j.jalgebra.2020.08.035}
  {\path{doi:10.1016/j.jalgebra.2020.08.035}}.

\bibitem[STA20]{STA20}
Chao Sun, Mehdi Tibouchi, and Masayuki Abe.
\newblock Revisiting the hardness of binary error {LWE}.
\newblock In {\em Information Security and Privacy -- ACISP 2020}, volume 12248
  of {\em Lecture Notes in Computer Science}, pages 425--444. Springer, 2020.
\newblock URL: \url{https://eprint.iacr.org/2020/666}, \href
  {https://doi.org/10.1007/978-3-030-55304-3_22}
  {\path{doi:10.1007/978-3-030-55304-3_22}}.

\bibitem[Ste24]{Steiner24}
Matthias~Johann Steiner.
\newblock The complexity of algebraic algorithms for {LWE}.
\newblock In {\em Advances in Cryptology -- EUROCRYPT 2024, Part III}, volume
  14653 of {\em Lecture Notes in Computer Science}, pages 375--403. Springer,
  2024.
\newblock URL: \url{https://eprint.iacr.org/2024/313}, \href
  {https://doi.org/10.1007/978-3-031-58734-4_13}
  {\path{doi:10.1007/978-3-031-58734-4_13}}.

\bibitem[Wag02]{Wagner02}
David Wagner.
\newblock A generalized birthday problem.
\newblock In {\em Advances in Cryptology -- CRYPTO 2002}, volume 2442 of {\em
  Lecture Notes in Computer Science}, pages 288--304. Springer, 2002.
\newblock \href {https://doi.org/10.1007/3-540-45708-9_19}
  {\path{doi:10.1007/3-540-45708-9_19}}.

\bibitem[WANT26]{WANT26}
Takeshi Wakao, Yusuke Aikawa, Shintaro Narisada, and Tsuyoshi Takagi.
\newblock Revisiting {ISD} algorithms and new decoding records for large weight
  syndrome decoding over {$\mathbb{F}_q$}.
\newblock In {\em Proceedings of the 12th International Conference on
  Information Systems Security and Privacy -- Volume 2: {ICISSP}}, pages
  13--24. SciTePress, 2026.
\newblock URL: \url{https://www.scitepress.org/Papers/2026/142486/142486.pdf},
  \href {https://doi.org/10.5220/0014248600004061}
  {\path{doi:10.5220/0014248600004061}}.

\bibitem[YZ24]{yamakawa2024verifiable}
Takashi Yamakawa and Mark Zhandry.
\newblock Verifiable quantum advantage without structure.
\newblock {\em Journal of the ACM}, 71(3):1--50, 2024.

\bibitem[ZvGZ{\etalchar{+}}26]{GDM26}
Daniel Zheng, Ingrid von Glehn, Yori Zwols, Iuliya Beloshapka, Lars Buesing,
  Daniel~M. Roy, Martin Wattenberg, Bogdan Georgiev, Tatiana Schmidt, Andrew
  Cowie, Fernanda Viegas, Dimitri Kanevsky, Vineet Kahlon, Hartmut Maennel,
  Sophia Alj, George Holland, Alex Davies, and Pushmeet Kohli.
\newblock {AI} co-mathematician: Accelerating mathematicians with agentic {AI},
  2026.
\newblock \href {https://arxiv.org/abs/2605.06651} {\path{arXiv:2605.06651}}.

\end{thebibliography}

\appendix
\crefalias{section}{appendix}

\section{\texorpdfstring{Satisfiability of average-case \FSubsetSum}
{Satisfiability of average-case F3-Subset-Sum}}\label{sec:fct:f3-subset-sum-existence}

In this section, we prove \Cref{fct:f3-subset-sum-existence}, which establishes the satisfiability phase transition of \Cref{prob:f3-subset-sum}. We emphasize that we do not nail down the low-order terms for the phase transition threshold, as it is not the focus of our paper.

\fctfsubsetsumexistence*
\begin{proof}
The first item follows from standard results in zero-sum theory~\cite{Olson69,diderrich1973combinatorial}. See also \cite[Section 5.5]{KOW26}. Here we focus on the other two.

For each nonempty $S\subseteq[m]$, let
\[
 I_S=\indicator\!\left[\sum_{i\in S}h_i=0\right],
 \qquad X=\sum_{\emptyset\ne S\subseteq[m]}I_S.
\]
The sum indexed by any fixed nonempty $S$ is uniform in $\F_3^n$, and hence
\[
 \E[I_S]=3^{-n},
 \qquad \E[X]=(2^m-1)\cdot3^{-n}.
\]
If $S$ and $T$ are distinct nonempty subsets, then
\(\sum_{i\in S}h_i\) and \(\sum_{i\in T}h_i\) are independent uniform vectors.  Thus the
indicators $I_S$ are pairwise independent and
\[
 \Var[X]=\sum_{\emptyset\ne S\subseteq[m]}\Var[I_S]\le\E[X].
\]

If $m\le(\kappa_\star-\eta)n$, Markov's inequality gives
\[
 \Pr[X>0]=\Pr[X\ge1]\le\E[X]=(2^m-1)\cdot3^{-n}\le2^{-\eta\cdot n}.
\]
If $m\ge(\kappa_\star+\eta)n$, then
\[
 \E[X]=(2^m-1)\cdot3^{-n}\ge2^{\eta\cdot n-1},
\]
and the Paley–Zygmund inequality gives
\[
 \Pr[X=0]\le\frac{\Var[X]}{\E[X^2]}\le\frac{\Var[X]}{\E[X]^2}
 \le\frac1{\E[X]}\le2^{1-\eta\cdot n}.
\]
These two bounds imply the claimed estimates.
\end{proof}

\section{\texorpdfstring{Uniqueness of \PointwiseBinaryErrorLWE over $\F_q$}
{Uniqueness of Pointwise Binary-Error LWE over Fq}}\label{sec:fct:lwe-existence}

In this section, we prove \Cref{fct:lwe-existence}, which establishes the uniqueness phase transition of \Cref{prob:pointwise-binary-error-lwe}. We emphasize that we do not nail down the low-order terms for the phase transition threshold, as it is not the focus of our paper.

For $q=3$, the same threshold was established by Bricout, Chailloux,
Debris-Alazard, and Lequesne~\cite[Section~6]{BCDL19}, though not in the pointwise
setting.

\fctlweexistence*
\begin{proof}
Fix $s\in\F_q^n$ and $e\in\{0,1\}^m$.  For every nonzero $u\in\F_q^n$, let
\[
 I_u=\indicator\!\left[A^\top u+e\in\{0,1\}^m\right],
 \qquad X=\sum_{u\ne0}I_u.
\]
The planted secret $s$ is the unique secret exactly when $X=0$.

Put $N=q^n-1$ and $p=(2/q)^m$.  For fixed $u\ne0$, each coordinate of
$A^\top u$ is uniform in $\F_q$. Hence
\[
 \E[I_u]=p,
 \qquad
 \E[X]=Np.
\]
If $u$ and $v$ are linearly independent, then $I_u$ and $I_v$ are independent.
For $v=c\cdot u$ where $c\in\F_q^*\setminus\{1\}$, we have $\E[I_uI_v]=\Pr[A^\top u=0]=q^{-m}$. Therefore
\[
 \Var[X]
 =N\cdot\left(p+(q-2)q^{-m}-(q-1)p^2\right)
 \le N\cdot\left(p+q^{-m+1}\right).
\]

Since $\kappa_q=\log_{q/2}(q)$, we have
\[
 \E[X]=(1-q^{-n})\cdot(q/2)^{\kappa_q n-m}.
\]
If $m\ge(\kappa_q+\eta)n$, Markov's inequality gives
\[
 \Pr[X>0]=\Pr[X\ge1]\le\E[X]
 \le(2/q)^{m-\kappa_q n}=2^{-\Omega_{q,\eta}(m)}.
\]
If $m\le(\kappa_q-\eta)n$, then the Paley-Zygmund inequality gives
\[
 \Pr[X=0]\le\frac{\Var[X]}{\E[X]^2}\le\frac1{Np}+\frac{q^{-m+1}}{Np^2}=\frac1{Np}\cdot\left(1+q\cdot2^{-m}\right)
 =2^{-\Omega_{q,\eta}(n)}.
\]
These two bounds imply the claimed estimates.
\end{proof}

\section{\texorpdfstring{Degree-$3$ XL experiments}{Degree-3 XL experiments}}
\label{app:relinearization-experiments}

We describe the counting prediction and experimental evidence for degree-$3$ XL over $\F_3$. The experiments below are separate from the proved
guarantees in \Cref{sec:lwe}.

We remark that higher-degree XL significantly blows up the system of equations and we cannot run enough experiments to draw convincing conclusions. Hence we focus on degree-$3$ XL here.

\paragraph{Degree-$3$ XL attack.}
Given an instance of the $q=3$ version of \Cref{prob:pointwise-binary-error-lwe}, form
\[
 f_i(x)=(L_i(x)-b_i)(L_i(x)-b_i+1),
 \qquad L_i(x)=\langle a_i,x\rangle.
\]
Degree-$3$ XL takes the original equations and their linear multiples 
$$
f_i, x_1f_i,x_2f_i,\ldots,x_nf_i,
$$
reduces them using $x_j^3=x_j$, and writes their coefficients in the basis of all
nonconstant reduced monomials $x^\alpha$ with
$\alpha\in\{0,1,2\}^n$ and $|\alpha|\le3$.  It then performs Gaussian elimination over
$\F_3$, with the constant coefficients moved to the right-hand side.
Full column rank is a sufficient recovery condition: it determines all linearized variables, and the degree-$1$ coordinates of this unique solution recover the secret. 

\paragraph{Predicted sample lower bound.}
Although each sample produces $n+1$ rows, these rows always satisfy a dependency.

\begin{lemma}\label{lem:degree3-xl-sample-rank}
There is a nontrivial linear dependence among $f_i,x_1f_i,\ldots,x_nf_i$.
\end{lemma}
\begin{proof}
Observe that modulo the field equations $x_j^3-x_j$, we have
\[
 (L_i-b_i-1)f_i=(L_i-b_i)^3-(L_i-b_i)=0,
\]
which is a nontrivial linear combination of $f_i,x_1f_i,\ldots,x_nf_i$ that always vanishes.
\end{proof}

To lower bound the number of samples $m$ needed for full column rank, we compute the number of nonconstant columns (i.e., reduced monomials of degrees $1$, $2$, and $3$) by
\[
 V(n)=n+\left(n+\binom n2\right)
           +\left(2\binom n2+\binom n3\right)
       =\frac{n^3+6n^2+5n}{6}.
\]
By \Cref{lem:degree3-xl-sample-rank}, $m$ samples supply at most $mn$ independent rows.
Hence the degree-$3$ XL attack requires at least
\begin{equation}
 m_{\mathrm{pred}}(n)
 =\left\lceil\frac{V(n)}n\right\rceil
 =\left\lceil\frac{n^2+6n+5}{6}\right\rceil
 =\frac{1+o(1)}6\cdot n^2
 \label{eq:degree3-xl-finite-prediction}
\end{equation}
samples to obtain a matrix of full column rank. This is a necessary condition for full column rank.

\paragraph{Experiments.}
We ran experiments on the degree-$3$ XL for every dimension $20\le n\le66$. The empirical outcome is surprisingly aligned with the theoretical lower bound prediction \eqref{eq:degree3-xl-finite-prediction}.
\Cref{tab:degree3-xl-experiments} lists a few data points from the experiments.

\begin{table}[H]
\centering
\small
\begin{tabular}{@{}rrc@{}}
\toprule
$n$ & $m_{\mathrm{pred}}$ & success rate\\
\midrule
20  & 88  & $100.0\%$ \\
23 & 112 & $78.8\%$ \\
24 & 121 & $99.7\%$ \\
25& 130 & $80.0\%$ \\
30 & 181 & $99.8\%$ \\
40& 308 & $100.0\%$ \\
50 & 468 & $100.0\%$ \\
60 & 661 & $100.0\%$  \\
61 & 682 & $69.0\%$  \\
66 & 793 & $100.0\%$ \\
\bottomrule
\end{tabular}
\caption{Representative degree-$3$ XL experiments.
``Success rate'' measures whether degree-$3$ XL successfully determines the unique secret given $m_{\mathrm{pred}}$ random samples. We ran $1000$ trials per $n$ for $n\leq 50$ and then $100$ trials per $n$ after.}
\label{tab:degree3-xl-experiments}
\end{table}

For all tested dimensions with
$n\not\equiv\pm1\pmod6$, the success rate at $m_{\mathrm{pred}}$ is above $99\%$ in our experiments. 
For all dimensions tested with $n\equiv\pm1\pmod6$, the rate was between $69\%$ and $86\%$, but a single additional sample always suffices to increase the success rate to $100\%$. This ``+1 surplus'' happens exactly when $n\cdot m_{\mathrm{pred}}=V(n)$ (i.e., no rounding in \eqref{eq:degree3-xl-finite-prediction}).

Based on the experiments, we make the following conjecture.

\begin{conjecture}\label{conj:deg3}
For $q=3$ and all sufficiently large $n$, degree-$3$ XL solves \Cref{prob:pointwise-binary-error-lwe} with probability at least $0.99$ whenever
$$
m\ge\begin{cases}
m_{\mathrm{pred}}(n) & n\not\equiv\pm1\pmod6,\\
m_{\mathrm{pred}}(n)+1 & n\equiv\pm1\pmod6.
\end{cases}
$$
\end{conjecture}

We emphasize that the techniques in this paper do not establish \Cref{conj:deg3}, not even asymptotically $m\ge(1+o(1))n^2/6$. Over fields of sufficiently large characteristic and size, techniques from \cite{nenashev2017note} give an analogous bound at $m=n^2/6+O(n)$. Over $\F_3$, \cite{SemaevTenti21} proves a related leading-form surjectivity bound for uniformly random quadratic forms.

\end{document}